\documentclass[12pt]{article}
\usepackage{amssymb,amsmath,amsthm,mathtools}

\usepackage{array,geometry}  
\usepackage{setspace}
\usepackage{bm,paralist,dsfont,nicefrac,bbm}

\usepackage{etoolbox,algorithm,algpseudocode}

\usepackage{tikz,float,tabularx}
\usetikzlibrary{shapes,arrows,patterns,decorations.pathreplacing}

\usepackage[driverfallback=hypertex,pagebackref=true,colorlinks]{hyperref}
\hypersetup{linkcolor=[rgb]{.7,0,0}}
\hypersetup{citecolor=[rgb]{0,.7,0}}
\hypersetup{urlcolor=[rgb]{.7,0,.7}}

\usepackage{xcolor} 

\usepackage[capitalize, nameinlink, noabbrev]{cleveref} 

\makeatletter
\newcommand{\algmargin}{\the\ALG@thistlm}
\algnewcommand{\parState}[1]{\State%
  \parbox[t]{\dimexpr\linewidth-\algmargin}{\strut #1\strut}}

\newtheorem{theorem}{Theorem}[section]
\newtheorem*{theorem*}{Theorem}

\newtheorem{definition}[theorem]{Definition}
\newtheorem*{definition*}{Definition}

\newtheorem{lemma}[theorem]{Lemma}
\newtheorem*{lemma*}{Lemma}

\newtheorem{claim}[theorem]{Claim}
\newtheorem*{claim*}{Claim}

\newtheorem*{fact*}{Fact}

\newtheorem{observation}[theorem]{Observation}
\newtheorem*{observation*}{Observation}

\newtheorem*{conjecture*}{Conjecture}

\newtheorem{corollary}[theorem]{Corollary}
\newtheorem*{corollary*}{Corollary}

\newtheorem*{remark*}{Remark}

\newtheorem{proposition}[theorem]{Proposition}
\newtheorem*{proposition*}{Proposition}

\crefname{claim}{Claim}{Claims}

\allowdisplaybreaks

\DeclarePairedDelimiter{\paren}{\lparen}{\rparen}

\DeclarePairedDelimiter{\bracket}{[}{]}
\DeclarePairedDelimiter{\set}{\lbrace}{\rbrace}

\DeclarePairedDelimiter{\co}{[}{)}

\DeclareMathOperator*{\argmax}{arg\,max}

\def\rev{\mathsf{rev}}

\def\1{\mathds{1}}

\def\diff{~\mathrm{d}}
\def\abstart{\mathsf{AB}\text{-}\mathsf{start}}

\def\vab{ \mathsf{V}_{ \mathtt{AB} } }
\def\vba{ \mathsf{V}_{ \mathtt{BA} } }

\def\util{\mathsf{util}}
\def\utilab{ \util_{ \mathtt{AB} } }
\def\utilba{ \util_{ \mathtt{BA} } }

\def\fp{\mathsf{FP}}

\title{Approximately Optimal Mechanism Design \\ for Competing Sellers}

\author{
Brendan Lucier\\
{\small Microsoft Research}
\and
Raghuvansh R. Saxena\\
{\small Tata Institute of Fundamental Research}
}

\date{} 
\begin{document}

\maketitle
\thispagestyle{empty}
\addtocounter{page}{-1}

\begin{abstract}
Two sellers compete to sell identical products to a single buyer. Each seller chooses an arbitrary mechanism, possibly involving lotteries, to sell their product. The utility-maximizing buyer can choose to participate in one or both mechanisms, resolving them in either order. Given a common prior over buyer values, how should the sellers design their mechanisms to maximize their respective revenues?

We first consider a Stackelberg setting where one seller (Alice) commits to her mechanism and the other seller (Bob) best-responds.  We show how to construct a simple and approximately-optimal single-lottery mechanism for Alice that guarantees her a quarter of the optimal monopolist's revenue, for any regular distribution. Along the way we prove a structural result: for any single-lottery mechanism of Alice, there will always be a best response mechanism for Bob consisting of a single take-it-or-leave-it price. We also show that no mechanism (single-lottery or otherwise) can guarantee Alice more than a $1/\mathrm{e}$ fraction of the monopolist revenue. Finally, we show that our approximation result does not extend to Nash equilibrium: there exist instances in which a monopolist could extract full surplus, but neither competing seller obtains positive revenue at any equilibrium choice of mechanisms.
\end{abstract}

\newpage
\pagenumbering{roman} 
\tableofcontents

\newpage
\setcounter{page}{1}
\pagenumbering{arabic}


\newpage


\section{Introduction}
\label{sec:intro}


Sellers nearly always face competitive pressure, especially in a digital economy where potential buyers can find identical products being sold by multiple sources.
In such oligopolistic environments, 
there are many ways that sellers can potentially compete, such as via price, quality, quantity, timing, and more~\cite{hicks1935annual,von1953theory,cournot2020researches,champsaur1989multiproduct}. 
Motivated by the theory of approximately optimal mechanism design, we consider yet another way that sellers could compete: by introducing randomness, in the form of lotteries, into their sales protocols.
In so doing, we view competing sellers as mechanism designers and consider how the class of mechanisms from which they choose might influence the outcome of perfect competition.

Imagine the following simple duopoly scenario: two identical sellers wish to sell identical products to a single potential buyer. The buyer can approach either of the sellers first, and in case he does not obtain the item from this seller, can then approach the other seller. The setting is Bayesian, with the buyer's value for the product drawn from a commonly known distribution.  How should revenue-maximizing sellers design their sales mechanisms? In the monopolist setting with only one seller, the theory of screening provides a classic solution: it is revenue-optimal to sell the good using a deterministic take-it-or-leave-it price~\cite{mussa1978monopoly,Mye81}.  In duopoly, however, a restriction to posted prices leads to classic Bertrand competition, with disastrous results for the sellers.  
Indeed, any price $p > 0$ posted by one seller would simply be undercut by the other, leading ultimately to revenue $0$ at equilibrium.\footnote{Throughout this paper we normalize the marginal cost of production to $0$.}
This outcome occurs both
at Nash equilibrium, where prices are chosen simultaneously, and also at  Stackelberg equilibrium, where one seller commits to their price first and the other seller responds.

The Bertrand analysis above applies when sellers compete on price, but in principle sellers need not restrict themselves to posted-price mechanisms.  For example, a seller could offer a \emph{lottery} in which the buyer pays a price for a \emph{chance} at obtaining the good.  It is known that even a monopolist seller may benefit from the use of lotteries for certain screening problems, such as when allocating multiple goods to a single buyer~\cite{HartN13,DaskalakisDT14,HartR15}.
Returning to the case of just a single good but multiple sellers, we ask: is the ability to offer random lotteries likewise helpful for generating revenue in a duopoly?
Can fully general mechanism designers in duopoly approximate, at equilibrium, the revenue obtainable by a monopolist?
\paragraph{Model}
To answer this question, we first require a formulation of sales protocols under competition.  In the setup we consider, each of two sellers can design an arbitrary sales mechanism.  There is a single buyer over which the sellers compete, with value drawn from a known prior.  After observing the proposed mechanisms and their realized value, the buyer can choose the order in which to interact with the sellers.  The buyer interacts with the selected first seller to the conclusion of their mechanism.  The buyer is then free to interact with the other seller, if desired, using that seller's chosen mechanism.  We note that Bertrand competition corresponds to the special case where each seller is restricted to posting a take-it-or-leave-it price.
%

Under this model, we note that the \emph{taxation principle} \cite{hammond1979straightforward,guesnerie1981taxation} applies: it is without loss for each seller to describe their mechanism as a \emph{pricing function} that assigns a price $p(x)$ to each probability of allocation $x \in [0,1]$.  Each seller's protocol then has the buyer select their preferred allocation probability, pay the corresponding price, and receive the good with the chosen probability.
One can therefore model our competitive scenario as a game between sellers in which the actions are pricing functions over lotteries.  Once the lottery menus have been selected, the buyer can first select any one entry from either one of the sellers.  Then, after the randomness of the chosen lottery has been resolved, the buyer is free to purchase from the other seller, if desired.



\paragraph{Results}
We begin our study by considering Stackelberg equilibria of this mechanism design game.  One seller (Alice) is the principal who can commit to her menu in advance, and the other seller (Bob) best-responds to the principal's choice.  In our main result, we show that the sellers can obtain positive revenue at Stackelberg equilibrium.  Moreover, the expected revenue obtained by the principal, Alice, is always at least a quarter of the expected revenue that she would obtain if she were a monopolist. This bound applies to any value distribution that satisfies either of two constraints: decreasing marginal revenue (DMR) or regularity (see Definition~\ref{def:regularity}).

\begin{theorem}
\label{thm:oneoptionalice-informal}
    Consider any instance of the duopoly game in which the buyer's value is drawn from a distribution $D$ that is either regular or DMR.  There exists a Stackelberg equilibrium in which the expected revenue of the principal is at least $\textsc{Rev}(D)/4$, where $\textsc{Rev}(D)$ is the maximum expected revenue obtainable by a monopolistic seller.
\end{theorem}

As it turns out, this $4$-approximation is achievable by a simple menu for Alice in which she offers only a single lottery to the buyer.  In our main technical lemma we show that, against any such single-lottery menu, Bob's best response is a posted price. This simplifies Alice's problem of optimizing over single-item menus, and we find that the optimal such menu always achieves the desired approximation bound.

\begin{lemma}[Key Lemma]
\label{lemma:fixedprice-informal}
    For any value distribution satisfying the conditions of \cref{thm:oneoptionalice-informal}, any single-lottery mechanism admits a best response that is a posted-price mechanism.
\end{lemma}

We also provide a lower bound: there are instances of the duopoly game in which Alice cannot obtain more than a $1/\mathrm{e}$ fraction of the monopolist revenue at any Stackelberg equilibrium.  

\begin{theorem}
\label{thm:lowerbound}
    {Consider an instance of the duopoly game in which the buyer's valuation distribution is a point mass.} 
    At every Stackelberg equilibrium, the expected revenue of the principal is at most $\textsc{Rev}(D)/\mathrm{e}$, where $\textsc{Rev}(D)$ is highest expected revenue obtainable by a monopolistic seller.
\end{theorem}

As partial progress toward closing the gap between our upper and lower bounds of $4$ and $\mathrm{e}$, 
we note a connection between these bounds and our approach of focusing on single-item menus by Alice and posted-price menus by Bob.  First, we show that if Alice is restricted to only using single-lottery mechanisms, then our $4$-approximation is tight.
On the other hand,
if Bob is restricted to only using posted-price mechanisms, 
then Alice can always obtain a $1/\mathrm{e}$ fraction of the monopolist revenue.
This latter approximation result applies to \emph{any} valuation distribution, even non-regular.

The fact that Alice can approximate the monopolist revenue at Stackelberg equilibrium raises a natural follow-up question: can such an approximation result also be supported at Nash equilibrium, where the sellers choose their mechanisms simultaneously and then the buyer interacts with them as before?  One observation is that our game always supports at least one pure Nash equilibrium, in which each seller simply offers the good for free.  But are there other equilibria with better revenue outcomes?
We show that there exist instances where the answer to this question is negative. Specifically, we provide instances in which $\textsc{Rev}(D)$ is positive and yet both sellers obtain $0$ revenue at every pure Nash equilibrium, and hence no non-trivial approximation to the monopolist revenue can be obtained. 
Furthermore, in these instances, the buyer's type distribution is a point mass meaning that the buyer's value is fixed and known to both sellers.

\begin{theorem}
\label{thm:nonash-informal}
    Consider an instance of the duopoly game in which the buyer's valuation distribution is a point mass.  At every pure Nash equilibrium, both sellers obtain revenue $0$. 
\end{theorem}

Taken together with the classic analysis of Bertrand competition, these results suggest that a duopolist can approximate a monopolist's revenue if they have (a) the use of general randomized mechanisms, plus (b) the power to commit as a first-mover. 
However, neither of these two -- randomization nor commitment -- is sufficient on its own to guarantee positive revenue.


\paragraph{Proof Techniques}
The key lemma that drives our main approximation result is that, against any single-lottery mechanism chosen by Alice, Bob's best response is a posted-price mechanism. 
This conclusion is reminiscent of Myerson's characterization of the optimal mechanism in monopoly,
but the duopoly setting introduces new challenges.  For one thing, recall that the buyer can choose which seller to purchase from first.  Depending on the menu selected by Bob, some buyer types may be better off buying from Alice first and others will prefer to buy first from Bob.  From Bob's perspective, this behavior introduces distortion into his perceived distribution over buyer values.  One might therefore be tempted to think of Bob as a monopolist who faces a distorted value distribution. Crucially, however, \emph{this distortion is endogenous to Bob's choice of mechanism}.  This endogeneity complicates the design problem faced by Bob.

To address this issue, we analyze the structure of buyer selections when Alice commits to a single-lottery mechanism. For any mechanism chosen by Bob, all buyer types can be grouped into three categories (see Section~\ref{sec:auxdist}). The lowest buyer types are unwilling to buy Alice's lottery, and will only consider Bob's mechanism.  Intermediate types are willing to purchase from either mechanism, but choose to interact with Bob first.  The highest types choose to purchase Alice's lottery first, and only if they lose do they consider buying from Bob.  The threshold between the second and third categories is endogenous to Bob's choice of mechanism.\footnote{Alice's lottery determines the threshold between the first and second categories.}

Our approach to Bob's design problem is therefore to fix an arbitrary threshold type $s$, then consider Bob's optimal mechanism subject to being consistent with this choice of $s$.  This fixes the distortion of buyer values perceived by Bob (see \cref{lemma:oneseller}), leaving us with a more traditional mechanism design problem but with an additional self-consistency constraint.  While the distorted value distribution may no longer be regular, a virtual value analysis allows us to conclude that Bob's optimal mechanism (a) has allocation supported on $\{0,x\}$, for some $x \geq 0$, for all types below the threshold $s$ (\cref{lemma:fixedprice-reduction}), and (b) generates revenue at most that of the best posted-price mechanism for types above the threshold $s$. It is here where we use our assumption that the value distribution is either regular or DMR: either of these conditions suffices to control the virtual values of the distorted distribution, which drives our analysis. This partial characterization is enough to let us bound the revenue of the optimal mechanism by a convex combination of posted-price mechanisms, from which we conclude that the optimal revenue (over all possible thresholds) is realized at some posted-price mechanism.

\subsection{Related Literature}

Our results relate to classic economic models of duopoly, including Cournot, Stackelberg, and Bertrand competition.  Beyond the perfect competition scenario we consider, a wide variety of models of imperfect competition and/or product differentiation (such as Hotelling models) have been studied along with the resulting impact on seller revenues at equilibrium. A full discussion is beyond our scope; see~\cite{jehle2001advanced} (Chapter 4) for a classic treatment of perfect and imperfect competition models, and~\cite{vives1999oligopoly} for a survey and overview of model variations.

Our problem also builds on the growing literature on approximately optimal screening mechanisms.  In the case of a monopolist seller, it is classic result that lotteries are not needed to sell a single item optimally~\cite{mussa1978monopoly,Mye81}.  In contrast, a more recent literature on screening problems with multiple items have shown that lotteries can be necessary to achieve optimal revenue~\cite{Thanassoulis04, ManelliV07, Pavlov11, HartN13, DaskalakisDT14, HartR15, DaskalakisDT17}.  This has motivated the design and analysis of simple and approximately optimal auctions \cite{LiuP18, CaiDW16, CaiZ17, EdenFFTW17b, EdenFFTW17a, BrustleCWZ17, DevanurW17, FuLLT18, BeyhaghiW19, CaiS21}.  From this literature, closest to our work is~\cite{cai2024bundling} which considers a competitive scenario where a multi-item seller competes against multiple single-item sellers for an additive buyer.  They focus on a Stackelberg solution concept, and establish bounds on the revenue achievable by the multi-item principal seller. 

  %

There is also an economic literature on competition between mechanism designers who compete to attract bidders from a potential pool.  Champsaur and Rochet~\cite{champsaur1989multiproduct} study a model of sellers who compete for a single buyer via quality differentiation, in the style of Mussa and Rosen~\cite{mussa1978monopoly}.  Interpreting vertical differentiation as product quality means that the competition is exclusive: the buyer will choose either one seller or the other, and never visit the sellers sequentially (as may happen in our model).  In a model with multiple bidders and sellers, \cite{mcafee1993mechanism} considers a multi-round game in which the sellers do not fully account for the externality that their mechanism choice imposes on other sellers.  This yields a symmetric equilibrium in which all sellers employ a second-price auction with reserve equal to production cost (in our case, $0$).  This analysis was later extended to a limit market as the number of buyers and sellers grows large~\cite{peters1997competition}.  When sellers compete for buyers by simultaneously setting reserve prices in competing second-price auctions, the resulting equilibrium need not force reserve prices down to the efficient level~\cite{burguet1999imperfect}.  
Our focus is on a small market with screening in which a single buyer can engage with one or more sellers in sequence, rather than having to choose between them in coordination with other potential buyers.





Price competition between sellers is also well-studied in the algorithmic mechanism design literature, with a particular focus on imperfect competition with combinatorial structure.  This includes sellers who compete to sell to a buyer with combinatorial valuation~\cite{babaioff2014price} as well as sellers of identical goods but for whom only certain agents have access~\cite{chawla2008bertrand,babaioff2013bertrand}.  
Designing a best-response mechanism is also related to mechanism design with outside options and/or participation costs~\cite{celik2009optimal,gonczarowski2024revenue}, albeit with endogenous rather than exogenous outside option values (determined by the competing mechanism).  {Nevertheless, as with exogenous outside options, one can view the best-response problem as a form of ``one-and-a-half dimensional'' mechanism design~\cite{fiat2016fedex,DevanurW17}, and our technical approach employs ideas from that literature.}

\subsection{Roadmap}

We define our model formally in Section~\ref{sec:prelim}.  In Section~\ref{sec:stackelberg} we analyze the Stackelberg model and prove Theorem~\ref{thm:oneoptionalice-informal}.  Our main technical lemma, Lemma~\ref{lemma:fixedprice-informal}, is proven in Sections~\ref{sec:keylemmaproof} through~\ref{sec:keylemma.proof.end}.  We prove our lower bound (Theorem~\ref{thm:lowerbound}) in Section~\ref{sec:stackelberg.lowerbound}, and in Section~\ref{sec:oneoptionbob} we show that this bound is achievable if Bob cannot use randomization in his mechanism.  In Section~\ref{sec:fixedtype} we analyze pure Nash equilibria and prove Theorem~\ref{thm:nonash-informal}.  We conclude and discuss future directions in Section~\ref{sec:conclusion}.

\section{Model and Preliminaries}
\label{sec:prelim}

Two sellers, Alice ($\textsc{A}$) and Bob ($\textsc{B}$), are each selling an identical copy of a non-divisible good for which they have no value.  There is a single buyer with value $v \geq 0$ for one copy of the good.
The value $v$ is drawn from publicly known distribution $D$.  We assume $D$ is atomless with well-defined density, and write $F$ and $f$ for its CDF and pdf respectively. We write $\Gamma_D(z) = z(1-F(z))$ for the (monopolist) revenue curve for $D$.  We assume that $D$ has a revenue-maximizing price, meaning that $\max_z \Gamma_D(z)$ is obtained at some finite $z$.\footnote{The restriction to continuous distributions is for technical convenience; all of our results extend to discrete distributions, under an appropriate discrete analog of our regularity conditions~\cite{elkind2007designing}.}


\paragraph{Mechanisms.} Alice and Bob can each choose a \emph{mechanism} with which to sell their good.  A mechanism takes as input a message $\sigma \in \Sigma$ from the buyer, where $\Sigma$ is an arbitrary message space, then returns an outcome $(a,t) \in \mathcal{O} = \{0,1\} \times \mathbb{R}_{\geq 0}$, where $a \in \{0,1\}$ indicates whether the buyer receives a copy of the good and $t$ represents a payment from the buyer to the seller.  A mechanism is then formally represented by a measurable mapping $\mathcal{M} \colon \Sigma \to \Delta( \mathcal{O} )$ with compact image.\footnote{Compactness ensures that every buyer has a well-defined utility-maximizing message.}  We assume for technical convenience that every mechanism accepts a null message $\emptyset$, equivalent to non-participation, upon which it always returns outcome $(0,0)$.


\paragraph{Timing.} Informally, the timing of our market proceeds as follows.  The two sellers first select their mechanisms, either simultaneously or sequentially depending on the equilibrium notion (described below).  The buyer then observes both mechanisms, as well as realized value $v$, and chooses one of the sellers to approach first.  The buyer interacts with the chosen mechanism and observes the outcome.  The buyer can then choose to interact with the remaining seller's mechanism (or not), after which the game ends.  More formally:  
\begin{enumerate}
    \item Alice and Bob select their mechanisms, $\mathcal{M}_A$ and $\mathcal{M}_B$, either simultaneously (for Nash equilibrium) or sequentially (for Stackelberg equilibrium).  
    \item The buyer observes $\mathcal{M}_A$ and $\mathcal{M}_B$ as well as the realization of value $v$. 
    \item The buyer selects a seller $i \in \{A,B\}$ to approach first.  Write $-i$ for the remaining seller.
    \item The buyer sends a message $\sigma_i$ to $\mathcal{M}_i$, yielding an outcome $(a_i, t_i) \sim \mathcal{M}_i(\sigma_i)$.  
    \item Next, the buyer sends a message $\sigma_{-i}$ to $\mathcal{M}_{-i}$, yielding an outcome $(a_{-i}, t_{-i}) \sim \mathcal{M}_{-i}(\sigma_{-i})$.\footnote{Recall that $\sigma_{-i}$ can be the null message $\emptyset$, interpreted as non-participation.}
\end{enumerate}
All players have quasi-linear utilities.  Alice obtains utility $t_A$, Bob obtains utility $t_B$, and the buyer's utility is $v \times \max\{a_A, a_B\} - t_A - t_B$.  Our solution concept is subgame perfect equilibrium.


\paragraph{Equilibrium.} For the buyer's equilibrium play, note first that if $a_i = 1$ then it is utility-maximizing for the buyer to not participate in $\mathcal{M}_{-i}$.  An undominated strategy for the buyer can therefore be described as a tuple $(i,\sigma_i,\sigma_{-i}) \in \{A,B\}\times\Sigma^2$, where $\sigma_{-i}$ is the message to send to $\mathcal{M}_{-i}$ conditional on the event that $a_i=0$.  The buyer's equilibrium condition is that $(i,\sigma_i,\sigma_{-i})$ maximizes expected buyer utility given $\mathcal{M}_A$ and $\mathcal{M}_B$.  
%
We will assume that the buyer breaks ties first in favor of maximizing their expected allocation from Bob, then in favor of maximizing their expected allocation from Alice.
This determines the expected allocations and transfers obtained at equilibrium given $\mathcal{M}_A$ and $\mathcal{M}_B$.  We can therefore define $\rev_{ \mathsf{A}, \paren*{ \mathcal{M}_A, \mathcal{M}_B } }\paren*{ D }$ to denote the expected payment the buyer makes to Alice when his valuation is sampled from $D$ and Alice's and Bob's mechanism are $\mathcal{M}_A$ and $\mathcal{M}_B$ respectively. Similarly, $\rev_{ \mathsf{B}, \paren*{ \mathcal{M}_A, \mathcal{M}_B } }\paren*{ D }$ will be the corresponding expected revenue for Bob.  Likewise, we define $x_{ \mathsf{A}, \paren*{ \mathcal{M}_A, \mathcal{M}_B } }\paren*{ D }$ and $x_{ \mathsf{B}, \paren*{ \mathcal{M}_A, \mathcal{M}_B } }\paren*{ D }$ for the expected allocation from Alice and Bob, respectively, given $\mathcal{M}_A$ and $\mathcal{M}_B$.

For the sellers' equilibrium play, we consider two different equilibrium concepts that correspond to subgame perfect equilibrium under two different timing models.  At \emph{(pure) Nash Equilibrium}, Alice and Bob choose their mechanisms simultaneously and in best response to each other.  Formally, the pair of mechanisms $(\mathcal{M}_A, \mathcal{M}_B)$ form a pure Nash equilibrium if $\mathcal{M}_A \in \argmax_{\mathcal{M}} \rev_{ \mathsf{A}, \paren*{\mathcal{M}, \mathcal{M}_B} }\paren*{ D }$ and $\mathcal{M}_B \in \argmax_{\mathcal{M}} \rev_{ \mathsf{B}, \paren*{\mathcal{M}_A, \mathcal{M}} }\paren*{ D }$.  We note that a pure Nash equilibrium is guaranteed to exist: for example, any pair of mechanisms that each have outcome $(1,0)$ in their image (i.e., giving away the item for free)  will form a pure Nash equilibrium.

At \emph{Stackelberg equilibrium}, Alice selects her mechanism first, and Bob then chooses his mechanism after observing Alice's choice.  Formally, Bob's strategy is a mapping $\mathbb{M}$ from Alice's choice of mechanism to a mechanism for Bob, and Alice's strategy is a mechanism $\mathcal{M}_A$.  These strategies form a Stackelberg equilibrium if (a) $\mathbb{M}(\mathcal{M}') \in \argmax_{\mathcal{M}}\rev_{ \mathsf{B}, \paren*{\mathcal{M}', \mathcal{M}} }\paren*{ D }$ for all mechanisms $\mathcal{M}'$ (Bob best-responds to any possible mechanism $\mathcal{M}'$ chosen by Alice) and $\mathcal{M}_A \in \argmax_{\mathcal{M}}\rev_{ \mathsf{A}, \paren*{\mathcal{M}, \mathbb{M}(\mathcal{M})} }\paren*{ D }$ (Alice is choosing the utility-maximizing mechanism in anticipation of Bob's best response).  In a slight abuse of notation, we will sometimes say that mechanisms $(\mathcal{M}_A, \mathcal{M}_B)$ are in Stackelberg equilibrium if $\mathcal{M}_B = \mathbb{M}(\mathcal{M}_A)$; i.e., when $\mathcal{M}_B$ is the on-path mechanism chosen by Bob at Stackelberg equilibrium.


\paragraph{Pricing Functions and the Taxation Principle.} 
We observe an analogue of the taxation principle: any mechanism can equivalently be defined by a pricing function.
Formally, a \emph{pricing} mechanism is described by a pricing function $M \colon [0,1] \to \mathbb{R}_{\geq 0} \cup \{\infty\}$.  The input to the mechanism is a choice of $x \in [0,1]$, which results in payment $M(x)$ and the buyer receiving the good with probability $x$.  We say that a pricing mechanism is \emph{proper} if (a) $M(0) = 0$, (b) $M$ is weakly increasing, and (c) $M$ is continuous and weakly convex on interval $[0,\bar{x}]$, where $\bar{x} = \sup \{ x \colon M(x) < \infty \}$ is the maximum probability offered at a finite price.

\begin{observation}[Taxation Principle]
\label{obs:taxation}
    For any mechanism $\mathcal{M}$ for either seller there is a strategically equivalent proper pricing mechanism. In the corresponding pricing rule, $M(x)$ is the minimum expected payment over any distribution of buyer messages that results in an allocation with probability at least $x$ (or $M(x) = \infty$ if no such message exists).\footnote{That $M(0) = 0$ follows because of the assumed null message.  That $M$ is non-decreasing follows from free disposal.  That $M$ is continuous and weakly convex follows because, for any $x_1,x_2$ with $M(x_1), M(x_2) < \infty$ and any $\lambda \in (0,1)$, the buyer could randomize between whichever messages yield expected allocations $x_1$ and $x_2$ to obtain expected allocation $\lambda x_1+(1-\lambda)x_2$ at an expected price of $\lambda M(x_1)+(1-\lambda)M(x_2)$, and hence $M$ is equal to its own lower convex envelope.}
\end{observation}

See \cref{sec:prelim-app} for the proof of \cref{obs:taxation}. Given \cref{obs:taxation}, we will henceforth think of the sellers as choosing pricing functions.  We will tend to write $A \colon [0,1] \to \mathbb{R} \cup \{\infty\}$ to denote the price function selected by Alice, and similarly $B$ for the corresponding function for Bob.  We allow pricing functions to stand in for mechanisms in all of our notation; e.g., $\rev_{ \mathsf{A}, (A,B)}(D)$ denotes Alice's expected revenue when the sellers choose pricing functions $A$ and $B$ respectively. We will sometimes omit the pricing functions and/or value distribution from our notation when clear from context.



\paragraph{Classes of Mechanisms.} We say that a mechanism is a \emph{single-lottery mechanism} if it offers the buyer a single lottery option, which provides a probability $z \in [0,1]$ to obtain the item.  The expected payment from participating in the lottery is $pz$ for some $p \geq 0$, which can be interpreted as either a payment of $p$ conditional on winning the lottery or an up-front payment of $pz$ to participate in the lottery. The buyer can either accept this single lottery (paying the appropriate price), or obtain nothing at a payment of $0$.  This corresponds to a price function $M$ of the form $M(x) = px$ for all $x \in [0,z]$ and $M(x) = \infty$ for $x > z$.  We say this is the lottery mechanism with probability $z$ and price $p$.  In the special case where $z = 1$, we say that the mechanism is a \emph{fixed-price} (or \emph{posted-price}) mechanism, which corresponds to offering the item at a take-it-or-leave-it price of $p$. 
\begin{definition}
\label{def:fixedprice}
For $p \geq 0$, the fixed price mechanism $\fp_p$ with price $p$ is defined to be so that for all $x \in [0, 1]$, we have $\fp_p(x) = px$. We say that a mechanism $M$ is fixed price if $M = \fp_p$ for some $p \geq 0$, in which case we refer to $p$ as the price of mechanism $M$.
\end{definition}

\paragraph{Monopolist Revenue and Virtual Values}
We now recall Myerson's characterization of the revenue-optimal monopolist mechanism~\cite{Mye81}, which will be useful in our proofs. Specifically, we consider a single seller selling to a single buyer whose valuation $v \geq 0$ is sampled from a distribution $D$ with probability density function $f$. As above, a mechanism for this setting is represented by a proper pricing function $M : [0, 1] \to \mathbb{R} \cup \{\infty\}$.
Faced with this mechanism, the buyer will try to maximize his utility, which means that he will buy the item with probability $\mathtt{z}_M(v) = \argmax_z \paren*{ z v - M(z) }$ at a price of $M\paren*{ \mathtt{z}_M(v) }$.\footnote{Note that properness implies that $\mathtt{z}_M(v)$ is well-defined for all $v$ since the maximum is attained for some $z \in [0, 1]$.} We will also adopt the convention that $\mathtt{z}_M\paren*{ \infty } = 1$. This means the seller's revenue, denoted $\rev_M\paren*{ D }$, is
\[
\rev_M\paren*{ D } = \int_0^{ \infty } f(v) \cdot M\paren*{ \mathtt{z}_M(v) } \diff v .
\]
The function $\mathtt{z}_M(v)$ is often referred to as the allocation rule for $M$.  By Myerson's revenue characterization,  $\mathtt{z}_M(v)$ is non-decreasing and $\rev_M\paren*{ D }$ equals the expected virtual welfare, defined as $\varphi_D\paren*{ v } = v - \frac{ 1 - F(v) }{ f(v) }$.
\begin{proposition}[Revenue equals expected virtual welfare~\cite{Mye81}]
\label{prop:revvv}
For any distribution $D$ and any mechanism $M$, it holds that $\mathtt{z}_M(v)$ is monotone non-decreasing and
\[
\rev_M\paren*{ D } = \int_0^{ \infty } \mathtt{z}_M(v) \cdot \paren*{ v \cdot f(v) - \paren*{ 1 - F(v) } } \diff v .
\]
\end{proposition}

\begin{definition}\label{def:regularity}
For distribution $D$ with virtual welfare function $\phi_D$ and density function $f$:
\begin{itemize}
    \item $D$ is \emph{regular} if $\varphi_D(v)$ is non-decreasing.
    \item $D$ has \emph{diminishing marginal revenue} (DMR) if $\varphi_D(v)f(v)$ is non-decreasing.
\end{itemize}
\end{definition}

\section{The Stackelberg Setting}
\label{sec:stackelberg}


In this section, we consider the Stackelberg setting and prove \cref{thm:oneoptionalice-informal}. Fix 
distribution $D$ from which the buyer's valuations are drawn and recall that $\Gamma_D\paren*{ v } = v \cdot \paren*{ 1 - F(v) }$ is the revenue curve of $D$. Throughout, we let $\mathtt{v} = \argmax_v \Gamma_D(v)$ be the Myerson price of the distribution $D$. We show that \cref{thm:oneoptionalice-informal} holds even if Alice's mechanism is a single lottery mechanism.
\begin{theorem}[Formal version of \cref{thm:oneoptionalice-informal}]
\label{thm:oneoptionalice-formal}
Let $D$ be a distribution that satisfies either regularity or DMR and let $\mathtt{v} = \argmax_v \Gamma_D(v)$ be the corresponding Myerson price\footnote{We break ties in favor of larger values.}. Let $p \leq \mathtt{v}$ be the smallest such that\footnote{$p$ is guaranteed to exist since $\Gamma_D\paren*{ 0 } = 0$ and any discontinuities in $\Gamma_D$ must be downward discontinuities. Together with our choice of the smallest $p$, this also implies that $\Gamma_D\paren*{ q } \leq \Gamma_D\paren*{ p }$ for all $0 \leq q \leq p$.} $\Gamma_D\paren*{ p } = \frac{ \Gamma_D\paren*{ \mathtt{v} } }{2}$ and suppose Alice's mechanism $A$ is a single lottery, of price $p$, that allocates with probability $z = \frac{1}{2}$. Then there exists a best-response mechanism $B$ for Bob 
such that
\[
\rev_{ \mathsf{A}, \paren*{ A, B } }\paren*{ D } \geq \frac{ \Gamma_D\paren*{ \mathtt{v} } }{4} \hspace{1cm}\text{and}\hspace{1cm} \rev_{ \mathsf{B}, \paren*{ A, B } }\paren*{ D } \geq \frac{ \Gamma_D\paren*{ \mathtt{v} } }{2} .
\]
Moreover, by decreasing $p$ marginally, we can approach the above inequalities for \emph{any} choice of best-response mechanism for Bob.
\end{theorem}

Observe that \cref{thm:oneoptionalice-formal} implies that the best Stackelberg equilibrium for Alice gives her revenue at least $\frac{ \Gamma_D\paren*{ \mathtt{v} } }{4}$. Note also that while \cref{thm:oneoptionalice-formal} establishes a revenue guarantee for Bob when Alice chooses the described mechanism, it does not rule out the possibility that there are mechanisms for Alice that give her a higher revenue at the cost of giving Bob a significantly lower revenue when he is best-responding. Establishing a lower bound on Bob's revenue at exact Stackelberg equilibrium is left as an open question.

The key lemma in our proof of \cref{thm:oneoptionalice-formal} is \cref{lemma:fixedprice-informal}, which is a structural result about Bob's best response mechanism. It says that when Alice uses a single lottery mechanism, then Bob has a best response that is a fixed price (take-it-or-leave-it) mechanism.

\begin{lemma}[Formal version of \cref{lemma:fixedprice-informal}]
\label{lemma:fixedprice-formal}
Let $A$ be a single lottery mechanism for Alice with probability $z \in [0, 1]$ and price $p \geq 0$. For all mechanisms $B$ for Bob, there exists $q \geq 0$ for which:
\[
\rev_{ \mathsf{B}, \paren*{ A, B } }\paren*{ D } \leq \rev_{ \mathsf{B}, \paren*{ A, \fp_q } }\paren*{ D } .
\]
Moreover, if $\max_{ \hat{p} \leq p } \rev_{ \mathsf{B}, \paren*{ A, \fp_{ \hat{p} } } }\paren*{ D } < \rev_{ \mathsf{B}, \paren*{ A, \fp_{ \mathtt{v} } } }\paren*{ D }$ and there exists $v \geq p$ that pays $ < pz$ to Alice, then the inequality is strict.
\end{lemma}
Note that if the Myerson price $\mathtt{v} = \argmax_v \Gamma_D(v)$ of the distribution $D$ is at most $p$, Bob can just use a fixed price mechanism with price $\mathtt{v}$ and get the optimal monopolist revenue.\footnote{Since $\mathtt{v} \leq p$, no buyer would visit Alice's mechanism before Bob's.} Thus, it suffices to show \cref{lemma:fixedprice-formal} in the case $\mathtt{v} > p$ and we make this assumption for the rest of this section. Before proving \cref{lemma:fixedprice-formal}, we show how it can be used to prove \cref{thm:oneoptionalice-formal}.
\begin{proof}[Proof of \cref{thm:oneoptionalice-formal}]
As the theorem is trivial otherwise, we can assume without loss of generality that $\Gamma_D\paren*{ \mathtt{v} } > 0$ implying that $0 < p < \mathtt{v}$. For all $q \geq 0$, consider the revenue obtained by Bob when he responds with the mechanism $\fp_q$. As formalized in \cref{lemma:revfixedprice} later, if $q \leq p$ then responding with $\fp_q$ makes all buyers purchase from Bob, implying that he gets revenue $\Gamma_D(q)$. On the other hand, if $q > p$, then Alice is offering a cheaper price and the buyer only purchases from Bob if they do not get the item from Alice, which happens with probability $1 - z = \frac{1}{2}$, implying that Bob's revenue is $\frac{ \Gamma_D(q) }{2}$. Combining, we get:
\[
\rev_{ \mathsf{B}, \paren*{ A, \fp_q } }\paren*{ D } = \begin{cases}
\Gamma_D(q) , &\text{~if~} q \leq p \\
\frac{ \Gamma_D(q) }{2} , &\text{~if~} q > p 
\end{cases} .
\]
As $\Gamma_D\paren*{ q } \leq \Gamma_D\paren*{ p }$ for all $0 \leq q \leq p$, Bob's revenue in the first case is at most $\Gamma_D\paren*{ p } = \frac{ \Gamma_D\paren*{ \mathtt{v} } }{2}$. Using this and the definition of $\mathtt{v}$, observe that setting $q = \mathtt{v}$ maximizes Bob's revenue. From \cref{lemma:fixedprice-formal}, observe that setting $B = \fp_{ \mathtt{v} }$ maximizes $\rev_{ \mathsf{B}, \paren*{ A, B } }\paren*{ D }$. We prove the theorem for this mechanism $B$. That $\rev_{ \mathsf{B}, \paren*{ A, B } }\paren*{ D } \geq \frac{ \Gamma_D\paren*{ \mathtt{v} } }{2}$ is immediate from the foregoing equation, and we only show that $\rev_{ \mathsf{A}, \paren*{ A, B } }\paren*{ D } \geq \frac{ \Gamma_D\paren*{ \mathtt{v} } }{4}$.

For this, we show that any buyer with valuation $v \geq p$ pays $p/2$ to Alice. When $v < \mathtt{v}$, this is because the buyers only purchase from Alice and pay her $pz = p/2$. When $\mathtt{v} \leq v$, this is because the buyer purchase from Alice first before purchasing from Bob (as $p < \mathtt{v}$ implies that she is offering the item for cheaper) and she will get revenue $pz = p/2$. Thus, we have: 
\[
\rev_{ \mathsf{A}, \paren*{ A, B } }\paren*{ D } \geq \frac{p}{2} \cdot \paren*{ 1 - F\paren*{ p } } = \frac{ \Gamma_D\paren*{ p } }{ 2 } = \frac{ \Gamma_D\paren*{ \mathtt{v} } }{ 4 } .
\]

For the ``moreover'' part, note that if $p$ is decreased marginally, then the ``moreover'' part of \cref{lemma:fixedprice-formal} implies that any best-response $B$ for Bob satisfies that any buyer with type $v \geq p$ pays at least $p/2$ to Alice, and the result follows. 
\end{proof}

We also note that the factor of $4$ obtained above is tight over the choice of single-lottery mechanism for Alice.  This is true even if we don't restrict Bob to using a fixed-price mechanism.

\begin{theorem}
\label{thm:1/4tight}
Let the distribution $D$ be a point mass at $1$ and let $A$ be a single lottery mechanism for Alice with probability $z \in [0, 1]$ and price $p \geq 0$. For all mechanisms $B$ for Bob maximizing $\rev_{ \mathsf{B}, \paren*{ A, B } }\paren*{ D }$, it holds that $\rev_{ \mathsf{A}, \paren*{ A, B } }\paren*{ D } \leq \frac{1}{4}$.
\end{theorem}
\begin{proof}
Observe that, if Bob uses a fixed price mechanism with price $p$, he ensures that the buyer buys from him and he earns a revenue of $p$. This means that $\max_B \rev_{ \mathsf{B}, \paren*{ A, B } }\paren*{ D } \geq p$. Let $B$ be an arbitrary maximizer of Bob's revenue. If the buyer goes to Alice first when Bob uses the mechanism $B$, he only purchases from Bob with probability at most $1 - z$. This means that $1 - z \geq \rev_{ \mathsf{B}, \paren*{ A, B } }\paren*{ D } \geq p$. As Alice's revenue is upper bounded by $pz$, this means that $\rev_{ \mathsf{A}, \paren*{ A, B } }\paren*{ D } \leq pz \leq z - z^2 \leq \frac{1}{4}$, as desired. On the other hand, if the buyer goes to Bob first when he uses the mechanism $B$, then as formalized in \cref{lemma:typesoftypes:d} below, Bob's bang-per-buck price must be at most $p$. This means that the only way we have $\rev_{ \mathsf{B}, \paren*{ A, B } }\paren*{ D } \geq p$ is if the buyer purchases from Bob with probability $1$. However, if this happens, the buyer never purchases from Alice, and we have $\rev_{ \mathsf{A}, \paren*{ A, B } }\paren*{ D } = 0$ finishing the proof.
\end{proof}

\subsection{Best-Responding to a Single Lottery Mechanism: \texorpdfstring{\cref{lemma:fixedprice-formal}}{}}
\label{sec:keylemmaproof}

We now turn to proving our key technical result, \cref{lemma:fixedprice-formal}, which says that Bob's best response to a single lottery mechanism is a fixed price mechanism.  Omitted proof details and auxilliary lemmas appear in \cref{app:stackelberg}.  Throughout this section, we fix a single lottery mechanism $A$ for Alice with allocation probability $z$ and price $p$.  
We write $a = pz$ for the ex ante expected payment from this lottery. We assume that $z < 1$, as if $z = 1$, then Alice is using a fixed price mechanism and hence Bob's best response is to undercut Alice with a posted price, as in the Bertrand setting. We also assume that $z, p > 0$ as otherwise the lemma is trivial. Now, for any arbitrary mechanism $B$ for Bob, note that the buyer with type $v$ can either go to Alice first or to Bob first. Thus, we have the following two options:
\begin{enumerate}
\item \label{item:ba} \textbf{Purchasing from Bob first:} When a buyer purchases from Bob first, he picks the option corresponding to $x$ from Bob, for some $x \in [0, 1]$, and if he is not allocated the item (this happens with probability $1 - x$), may or may not purchase from Alice, depending on whether or not it improves his utility. This gives the buyer a utility of:
\begin{equation}
\label{eq:utilba}
\utilba\paren*{ B, x, v } = x v - B\paren*{ x } + \paren*{ 1 - x } \cdot \max\paren*{ 0, z v - a } .
\end{equation}
\item \label{item:ab} \textbf{Purchasing from Alice first:} When a buyer purchases from Alice first, he gets the item with probability $z$ at a price of $p$. In case he is not allocated (this happens with probability $1 - z$), he will pick the option corresponding to $x$ from Bob, for some $x \in [0, 1]$, giving him a total utility of:
\begin{equation}
\label{eq:utilab}
\utilab\paren*{ B, x, v } = z v - a + \paren*{ 1 - z } \cdot \paren*{ x v - B\paren*{ x } } .
\end{equation}
\end{enumerate}
Of course the buyer will pick the option that maximizes his utility. In other words, he will pick the option that maximizes:
\begin{equation}
\label{eq:util}
\util\paren*{ B, x, v } = \max\paren*{ \utilba\paren*{ B, x, v }, \utilab\paren*{ B, x, v } } .
\end{equation}
The buyer breaks ties in favor of buying from Bob first, then in favor of larger values of $x$.
Having fixed the buyer's choice, we can partition the set of potential buyer valuations into two sets, $\vba\paren*{ B }$, $\vab\paren*{ B }$, based on whether the a buyer with that valuation buys from Bob or Alice first, respectively.
An important property about the set $\vab\paren*{ B }$ is that it if a type $v \in \vab\paren*{ B }$, then all higher types are also in $\vab\paren*{ B }$.
\begin{lemma}[Follows from \cref{lemma:typesoftypes:d,lemma:typesoftypes:e}]
\label{lemma:vab}
For all $0 \leq v \leq v'$, $v \in \vab\paren*{ B } \implies v' \in \vab\paren*{ B }$. 
\end{lemma}
In other words, buyer behavior is monotone: low types buy from Bob first and high types buy from Alice first.
Define $\abstart\paren*{ B } = \inf\paren*{ \vab\paren*{ B } }$ with the convention that $\abstart\paren*{ B } = \infty$ if $\vab\paren*{ B } = \emptyset$. That is, $\abstart\paren*{ B }$ is the threshold type that separates those who select Bob first from those that select Alice first. Note from \cref{lemma:vab} that $v > \abstart\paren*{ B }$ implies $v \in \vab\paren*{ B }$. Also note from our tie-breaking that if $v \leq p$ then the buyer never purchases from Alice first, so we must have $\abstart\paren*{ B } \geq p$. Finally, as shown \cref{lemma:typesoftypes:h}, our tie-breaking also implies that $\abstart\paren*{ B } \in \vba\paren*{ B }$ as long as $\abstart\paren*{ B } < \infty$.

\subsection{An Auxiliary Distribution}
\label{sec:auxdist}

To characterize Bob's best-response mechanism, we next define an auxiliary distribution that effectively reduces the duopoly setting we consider to a monopolist setting. 
We use $\mathtt{x}_B(v)$ to denote the allocation chosen from Bob's mechanism by a buyer of type $v$, given Alice's single-lottery mechanism. We adopt the convention $\mathtt{x}_B\paren*{ \infty } = 1$.
Recall that, in contrast, $\mathtt{z}_B(v)$ denotes the expected allocation that would be chosen from Bob's mechanism if Bob were a monopolist.

Let $v$ be the buyer's type and consider the buyer's optimization problem in the following cases:
\begin{itemize}
\item \textbf{When $v < p$:} As we established that $\abstart\paren*{ B } \geq p$, we have that these buyers lie in $\vba\paren*{ B }$ and will choose the option $\mathtt{x}_B(v)$ to maximize \cref{eq:utilba}. When $v < p$, this means that they will choose $\mathtt{x}_B(v)$ to be the maximizer of $x v - B\paren*{ x }$. Observe that this is exactly $\mathtt{z}_B(v)$, the option that would be chosen by the buyer if Bob were the only seller.
\item \textbf{When $p \leq v \leq \abstart\paren*{ B }$:} As explained above, these types also lie in $\vba\paren*{ B }$ and will choose the option $\mathtt{x}_B(v)$ to maximize \cref{eq:utilba}. As $v \geq p$, this means that they will choose $\mathtt{x}_B(v)$ to be the maximizer of $x \cdot \paren*{ a + \paren*{ 1 - z } \cdot v } - B\paren*{ x }$. Observe that this is exactly $\mathtt{z}_B\paren*{ a + \paren*{ 1 - z } \cdot v }$, the option that would be chosen by a buyer with a modified valuation $a + \paren*{ 1 - z } \cdot v$ if Bob were the only seller.
\item \textbf{When $v > \abstart\paren*{ B }$:} The types all lie in $\vab\paren*{ B }$ and will choose $\mathtt{x}_B(v)$ to maximize \cref{eq:utilab}. As $z < 1$, this means that they will choose $\mathtt{x}_B(v)$ to be the maximizer of $x v - B\paren*{ x }$. Observe that this is exactly $\mathtt{z}_B(v)$, the option chosen by the buyer if Bob were the only seller.  However, the revenue obtained by Bob has shrunk by a factor of $1 - z$: the buyer only comes to Bob if he did not get the item from Alice, which happens with probability $1 - z$. 
\end{itemize}

Combining these cases, we get that if we consider a different distribution where the density of each type $v \in [p, \abstart\paren*{ B }]$ is shifted to $a + \paren*{ 1 - z } \cdot v$ and the density of each type $v > \abstart\paren*{ B }$ is scaled down by a factor $1 - z$, then the revenue obtained by a monopolist seller with this distribution is the same as the revenue Bob makes in our duopoly. Next, we make this formal by defining, for all $s \geq p$ (think of $s$ as $\abstart\paren*{ B }$), an auxiliary distribution $D_s$ with the following probability density (recalling that $z < 1$):
\begin{equation}
\label{eq:fs}
f_s\paren*{ v } = \begin{cases}
f\paren*{ v } , &\text{~if~} v \leq p \\
\frac{1}{ 1 - z } \cdot f\paren*{ \frac{ v - a }{ 1 - z } } , &\text{~if~} p < v < a + \paren*{ 1 - z } \cdot s \\
0 , &\text{~if~} a + \paren*{ 1 - z } \cdot s \leq v \leq s \\
\paren*{ 1 - z } \cdot f\paren*{ v } , &\text{~if~} s < v < \infty \\
\end{cases} .
\end{equation}
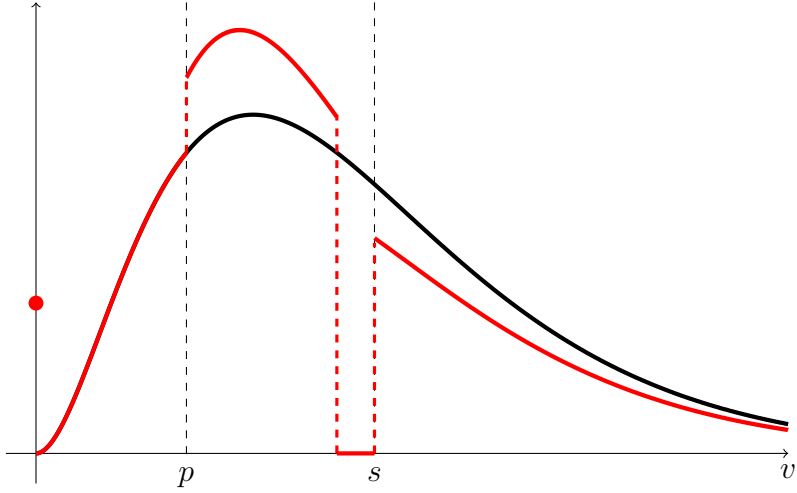
\begin{figure}
\centering
\begin{tikzpicture}[domain=0:10,samples=1000]
    \draw[->] (-0.4,0) -- (10,0) node[below] {$v$};
    \draw[->] (0,-0.4) -- (0,6);
    \draw[dashed] (2,6) -- (2,-0.04) node[below] {$p$};
    \draw[dashed] (4.5,6) -- (4.5,-0.04) node[below] {$s$};
    \draw[black, ultra thick] plot (\x,{pow(2, 4-\x) * pow(\x/2,2)});
    \draw[red,dashed,very thick] (2,4) -- (2,5);
    \draw[red, ultra thick] plot[domain=0:2] (\x,{(pow(2, 4-\x) * pow(\x/2,2))});
    \draw[red,dashed,very thick] (4,0) -- (4,4.5);
    \draw[red, ultra thick] plot[domain=2:4] (\x,{(\x <= 4) * (1.25 * pow(2, 4.5 -1.25 * \x) * pow(\x/1.6 - 0.25,2))});
    \draw[red,dashed,very thick] (4.5,0) -- (4.5,2.8);
    \draw[red, ultra thick] plot[domain=4:4.5] (\x,{0});
    \draw[red, ultra thick] plot[domain=4.5:10] (\x,{(0.8 * pow(2, 4-\x) * pow(\x/2,2))});
    \node[red, fill, circle, inner sep =2pt] at (0,2) {};
\end{tikzpicture}
\caption{Constructing the distribution $D_s$ ({\color{red}{red}}) from the distribution $D$ (black). The density below $p$ is unchanged, the density in the interval $\bracket*{ p, s }$ is squished into a smaller interval $\bracket*{ p, a + \paren*{ 1 - z } \cdot s }$, and the density above $s$ is scaled down by a factor of $1 - z$ with the remaining probability moved to an atom at $0$.}
\label{fig:ds}
\end{figure}
In order for this to be a well defined distribution, we additionally add a point mass of $z \cdot \paren*{ 1 - F(s) }$ at $0$, which corresponds to buyer types $v > s$ who visit Alice first and (with probability $z$) acquire an item from her and hence never visit Bob's mechanism.  As mentioned above, the distribution $D_s$ has the following property, whose proof we defer.

\begin{lemma}
\label{lemma:oneseller}
Let $B$ be a mechanism for Bob and $s = \abstart\paren*{ B }$. For all $v \geq 0$, we have:
\begin{equation}
\label{eq:oneseller}
\mathtt{x}_B(v) = 
\begin{cases}
\mathtt{z}_B(v) , &\text{~if~} 0 \leq v < p \\
\mathtt{z}_B\paren*{ a + \paren*{ 1 - z } \cdot v } , &\text{~if~} p \leq v \leq s \\
\mathtt{z}_B(v) , &\text{~if~} s < v 
\end{cases} .
\end{equation}
Additionally, we have:
\[
\rev_B\paren*{ D_s } = \rev_{ \mathsf{B}, \paren*{ A, B } }\paren*{ D } .
\]
\end{lemma}

\subsection{Special Case of \texorpdfstring{\cref{lemma:fixedprice-formal}}{}: Bottom Proper Mechanisms}

We now define a special class of mechanisms that we call \emph{bottom proper}, and prove \cref{lemma:fixedprice-formal} under the assumption that Bob's mechanism $B$ is bottom proper. Roughly speaking, a mechanism $B$ is bottom proper if there is some lottery in the mechanism $B$ such that any buyer with a type $v \in \vba\paren*{ B }$ either purchases that specific lottery or does not purchase at all. Formally, we define:
\begin{definition}
\label{def:bottomproper}
A mechanism $B$ for Bob is bottom proper if there exists $x \in [0, 1]$ such that, defining $\hat{p} = 0$ if $x = 0$ and $\hat{p} = \frac{ B\paren*{ x } }{ x }$ otherwise, we have $\hat{p} \leq p$ and for all $v \in \vba\paren*{ B }$, we have $\mathtt{x}_B(v) = x \cdot \1\paren*{ v \geq \hat{p} }$.
\end{definition}
We now show a lemma that will be helpful in showing \cref{lemma:fixedprice-formal} for the case of bottom proper mechanisms.
\begin{lemma}
\label{lemma:fixedprice-bottomproper}
Let $A$ be a single lottery mechanism for Alice with probability $z \in [0, 1]$ and price $p \geq 0$. For any bottom proper mechanism $B$, there exists $q \geq 0$ for which:
\[
\rev_{ \mathsf{B}, \paren*{ A, B } }\paren*{ D } \leq \rev_{ \mathsf{B}, \paren*{ A, \fp_q } }\paren*{ D } .
\]
Moreover, letting $x$ and $\hat{p}$ be as promised by \cref{def:bottomproper}, we have that if $\rev_{ \mathsf{B}, \paren*{ A, \fp_{ \hat{p} } } }\paren*{ D } < \rev_{ \mathsf{B}, \paren*{ A, \fp_{ \mathtt{v} } } }\paren*{ D }$, then the inequality is strict unless $x = 0$.
\end{lemma}
\begin{proof}
Fix an arbitrary single lottery mechanism $A$ for Alice with probability $z$ and price $p$ and let $a = pz$. Also fix a bottom proper mechanism $B$ for Bob and let $s = \abstart\paren*{ B }$ and $x$ and $\hat{p} \leq p$ be the value promised by \cref{def:bottomproper}. Using our auxiliary distribution, specifically \cref{lemma:oneseller}, to consider an equivalent setting with only seller, \cref{prop:revvv} implies
\[
\rev_{ \mathsf{B}, \paren*{ A, B } }\paren*{ D } = \rev_B\paren*{ D_s } = \int_0^{ \infty } \mathtt{z}_B(v) \cdot \paren*{ v \cdot f_s(v) - \paren*{ 1 - F_s(v) } } \diff v . 
\]
To analyze this integral, first consider the values $0 \leq v < \hat{p}$. As $\hat{p} \leq p$, for these value, the buyer purchases the same option from Bob regardless of whether or not Alice is present and we have $\mathtt{z}_B(v) = \mathtt{x}_B(v)$ (see \cref{lemma:oneseller}) From the definition of bottom proper, we get that $\mathtt{z}_B(v) = 0$ implying that these terms vanish. We get:
\[
\rev_{ \mathsf{B}, \paren*{ A, B } }\paren*{ D } = \int_{ \hat{p} }^{ \infty } \mathtt{z}_B(v) \cdot \paren*{ v \cdot f_s(v) - \paren*{ 1 - F_s(v) } } \diff v . 
\]
Next, we consider the values $\hat{p} \leq v \leq a + \paren*{ 1 - z } \cdot s$. Proceeding similarly as above, for these values, we get from \cref{lemma:oneseller} that $\mathtt{z}_B(v) = \mathtt{x}_B(v')$ for some $\hat{p} \leq v' \leq s$. It follows that $v' \in \vba\paren*{ B }$ and the definition of bottom proper (\cref{def:bottomproper}) implies that $\mathtt{z}_B(v) = x$. This means that:
\begin{multline*}
\rev_{ \mathsf{B}, \paren*{ A, B } }\paren*{ D } = x \cdot \int_{ \hat{p} }^{ a + \paren*{ 1 - z } \cdot s } \paren*{ v \cdot f_s(v) - \paren*{ 1 - F_s(v) } } \diff v \\
+ \int_{ a + \paren*{ 1 - z } \cdot s }^{ \infty } \mathtt{z}_B(v) \cdot \paren*{ v \cdot f_s(v) - \paren*{ 1 - F_s(v) } } \diff v . 
\end{multline*}
Let $\mathtt{v} = \argmax_v \Gamma_D(v)$ be the Myerson price for the distribution $D$. We now consider the values $a + \paren*{ 1 - z } \cdot s < v < \max\paren*{ s, \mathtt{v} }$ and show that the value of $v \cdot f_s(v) - \paren*{ 1 - F_s(v) }$ for these values of $v$ is non-positive. Indeed, for $a + \paren*{ 1 - z } \cdot s < v \leq s$, this is because the distribution $D_s$ has no probability mass on these points (see \cref{eq:fs}) and for $s < v < \max\paren*{ s, \mathtt{v} }$, this is because, for these values, we have that both $f_s\paren*{ \cdot }$ and $1 - F\paren*{ \cdot }$ are just $f\paren*{ \cdot }$ and $1 - F\paren*{ \cdot }$ scaled down by a factor of $1 - z$ and we are done as $v < \mathtt{v}$. Having shown non-positivity, note that reducing $\mathtt{z}_B(v)$ will only increase the value of the expression. In particular, because of monotonicity of $\mathtt{z}_B\paren*{ \cdot }$ (see \cref{prop:revvv}), we can reduce it to $x$ and we get:
\begin{multline*}
\rev_{ \mathsf{B}, \paren*{ A, B } }\paren*{ D } \leq x \cdot \int_{ \hat{p} }^{ \max\paren*{ s, \mathtt{v} } } \paren*{ v \cdot f_s(v) - \paren*{ 1 - F_s(v) } } \diff v \\
+ \int_{ \max\paren*{ s, \mathtt{v} } }^{ \infty } \mathtt{z}_B(v) \cdot \paren*{ v \cdot f_s(v) - \paren*{ 1 - F_s(v) } } \diff v . 
\end{multline*}
Finally, for the values $v > \max\paren*{ s, \mathtt{v} }$, a similar argument as above shows that the factor $v \cdot f_s(v) - \paren*{ 1 - F_s(v) }$ is non-negative and increasing $\mathtt{z}_B(v)$ to $1$ will only increase the value of the expression. Using this and fact that the anti-derivative of this factor is just the negative of the revenue curve $\Gamma_s$ for the distribution $D_s$, we get:
\[
\rev_{ \mathsf{B}, \paren*{ A, B } }\paren*{ D } \leq x \cdot \paren*{ \Gamma_s\paren*{ \hat{p} } - \Gamma_s\paren*{ \max\paren*{ s, \mathtt{v} } } } + \paren*{ 1 - z } \cdot \Gamma_D\paren*{ \max\paren*{ s, \mathtt{v} } } .
\]
To finish, we 
note that a direct analysis of the revenue curve for distribution $D_s$ and the revenue obtained by any fixed-price mechanism (see \cref{lemma:gammas,lemma:revfixedprice}) 
implies that $\Gamma_s\paren*{ \hat{p} } \leq \Gamma_D\paren*{ \hat{p} }$, where the latter is just Bob's revenue when he uses the mechanism $\fp_{ \hat{p} }$. Also, $\Gamma_s\paren*{ \max\paren*{ s, \mathtt{v} } } = \paren*{ 1 - z } \cdot \Gamma_D\paren*{ \max\paren*{ s, \mathtt{v} } }$ is Bob's revenue when he uses the mechanism $\fp_{ \max\paren*{ s, \mathtt{v} } }$. This means that:
\[
\rev_{ \mathsf{B}, \paren*{ A, B } }\paren*{ D } \leq x \cdot \rev_{ \mathsf{B}, \paren*{ A, \fp_{ \hat{p} } } }\paren*{ D } + \paren*{ 1 - x } \cdot \rev_{ \mathsf{B}, \paren*{ A, \fp_{ \max\paren*{ s, \mathtt{v} } } } }\paren*{ D } .
\]
Thus, one of $q = \hat{p}$ or $q = \max\paren*{ s, \mathtt{v} }$ satisfies the required conditions of the lemma. For the ``moreover'' part, note that $\Gamma_s\paren*{ \max\paren*{ s, \mathtt{v} } } = \paren*{ 1 - z } \cdot \Gamma_D\paren*{ \max\paren*{ s, \mathtt{v} } }$ implies from \cref{lemma:revfixedprice} that $\Gamma_s\paren*{ \max\paren*{ s, \mathtt{v} } } \leq \rev_{ \mathsf{B}, \paren*{ A, \fp_{ \mathtt{v} } } }\paren*{ D }$. Thus, if we have $\rev_{ \mathsf{B}, \paren*{ A, \fp_{ \hat{p} } } }\paren*{ D } < \rev_{ \mathsf{B}, \paren*{ A, \fp_{ \mathtt{v} } } }\paren*{ D }$, then unless $x = 0$, setting $q = \mathtt{v}$ would give a strict inequality.
%
%
\end{proof}

\subsection{Completing the Proof of \texorpdfstring{\cref{lemma:fixedprice-formal}}{}}
\label{sec:keylemma.proof.end}

To finish the proof of \cref{lemma:fixedprice-formal}, we now show that for any mechanism $B$ for Bob, there exists a bottom proper mechanism $B'$ that generates at least as much revenue.
\begin{lemma}
\label{lemma:fixedprice-reduction}
Let $A$ be a single lottery mechanism for Alice with probability $z \in [0, 1]$ and price $p < \mathtt{v}$. For all mechanisms $B$ for Bob, there exists a bottom proper mechanism $B'$ for Bob for which:
\[
\rev_{ \mathsf{B}, \paren*{ A, B } }\paren*{ D } \leq \rev_{ \mathsf{B}, \paren*{ A, B' } }\paren*{ D } .
\]
Moreover, if there exists $v \in \vba\paren*{ B }$ with $\mathtt{x}_B(v) > 0$, then $B'$ satisfies \cref{def:bottomproper} with $x = \mathtt{x}_B(s)$, where $s = \abstart\paren*{ B }$.
\end{lemma}

Before proving \cref{lemma:fixedprice-reduction}, we show that \cref{lemma:fixedprice-bottomproper,lemma:fixedprice-reduction} together imply \cref{lemma:fixedprice-formal}.
\begin{proof}[Proof of \cref{lemma:fixedprice-formal} assuming \cref{lemma:fixedprice-bottomproper,lemma:fixedprice-reduction}]
As the rest is straightforward, we only show the ``moreover'' part. Assume that $\max_{ \hat{p} \leq p } \rev_{ \mathsf{B}, \paren*{ A, \fp_{ \hat{p} } } }\paren*{ D } < \rev_{ \mathsf{B}, \paren*{ A, \fp_{ \mathtt{v} } } }\paren*{ D }$ and that there exists $v \geq p$ that pays $ < pz$ to Alice. From the latter conclude that $v \in \vba\paren*{ B }$ with $\mathtt{x}_B(v) > 0$. From the monotonicity of $\mathtt{x}_B\paren*{ \cdot }$ (see \cref{lemma:typesoftypes:fg}), we conclude that $\mathtt{x}_B(s) > 0$, where $s = \abstart\paren*{ B }$. Using \cref{lemma:fixedprice-reduction}, this means that there exists a bottom proper mechanism $B'$ for Bob that satisfies \cref{def:bottomproper} with $x = \mathtt{x}_B(s) > 0$. We are now done using \cref{lemma:fixedprice-bottomproper}.
\end{proof}

We now prove \cref{lemma:fixedprice-reduction}.
\begin{proof}
Fix an arbitrary single lottery mechanism $A$ for Alice with probability $z$ and price $p$ and let $a = pz$. Also fix an arbitrary mechanism $B$ for Bob and let $s = \abstart\paren*{ B }$ for convenience. If $\mathtt{x}_B(v) = 0$ for all $v \in \vba\paren*{ B }$, \cref{def:bottomproper} with $x = 0$ says that $B$ is bottom proper and the lemma follows, so we assume otherwise. This assumption lets us define $\hat{p} = \sup_{ v \in \vba\paren*{ B } : \mathtt{x}_B(v) > 0 } \frac{ B\paren*{ \mathtt{x}_B(v) } }{ \mathtt{x}_B(v) }$. That is, $\hat{p}$ is the supremum of the ``bang-per-buck'' price paid by the buyer to Bob when he purchases from Bob first.
Ideally, as buyers with higher valuations would be willing to pay a higher bang-per-buck price, we should be able to say that the supremum in $\hat{p}$ is attained at the highest valuation $v \in \vba\paren*{ B }$, namely $s$. However, this requires some care in the corner case where $s = \infty$ and we verify this formally in \cref{lemma:typesoftypes:e,lemma:typesoftypes:i}, which allows us to conclude that 
$\hat{p} = \frac{ B\paren*{ \mathtt{x}_B(s) } }{ \mathtt{x}_B(s) }$.

Next note that since types who purchase from Bob first only take options with a lower bang-per-buck price than Alice's (single) lottery, we have that $\hat{p} \leq p$. Define the mechanism $B'$ as:
\begin{equation}
\label{eq:botproper}
B'\paren*{ x } = \max\paren*{ B\paren*{ x }, \hat{p} \cdot x } .
\end{equation}
We claim that $B'$ is the required mechanism. We do this in three steps:

\medskip\noindent
\textbf{Step 1: $B$ and $B'$ have the same threshold type $s$.}  We show first that $\abstart\paren*{ B' } = s$. For this, note that all prices in $B'$ are at least as high as the corresponding prices in $B$, which implies that the buyer's obtained utility can only be lower. In particular, if the buyer picks an option $x$ from Bob when his mechanism is $B$, and $x$ satisfies $B'(x) = B(x)$, then the buyer will pick the same option when Bob's mechanism is $B'$. From this, note that a buyer with type $s$ and any buyer with type larger than $s$ (which must be in $\vab\paren*{ B }$ and pay a bang-per-buck price larger than $p \geq \hat{p}$ to Bob) buys the same option from Bob in both $B$ and $B'$. It follows that $\vab\paren*{ B } = \vab\paren*{ B' }$ which implies that $\abstart\paren*{ B' } = s$.

\medskip\noindent
\textbf{Step 2: $B'$ is bottom proper.} We show that \cref{def:bottomproper} holds with $x = \mathtt{x}_B(s)$. For this, note that if $v \in \vba\paren*{ B' }$ is such that $v < \hat{p}$, then the buyer will never purchase from Bob as Bob's bang-per-buck price from \cref{eq:botproper} is always larger than $v$. It follows that $\mathtt{x}_{ B' }(v) = 0$. Next, consider $v \in \vba\paren*{ B' }$ such that $v \geq \hat{p}$. For such a type $v$, \cref{def:bottomproper} requires us to show that $\mathtt{x}_{ B' }(v) = \mathtt{x}_B(s)$. As a buyer with type $s$ buys the same option from Bob in both $B$ and $B'$, we have $\mathtt{x}_B(s) = \mathtt{x}_{ B' }(s)$ and it suffices to show that $\mathtt{x}_{ B' }(v) = \mathtt{x}_{ B' }(s)$. 

For this, note that from the definition of $\mathtt{x}_{ B' }(v)$ that a buyer with type $v$ gets at least as much utility buying the option $\mathtt{x}_{ B' }(v)$ from Bob as he gets when buying the option $\mathtt{x}_{ B' }(s)$. However, the definition of $B'$ says that Bob's bang-per-buck price for the option $\mathtt{x}_{ B' }(v)$ is at least $\hat{p}$, which is his bang-per-buck price with option $\mathtt{x}_{ B' }(s)$. Thus, the only way this is possible is if the probability of sale is also larger, {\em i.e.}, we have $\mathtt{x}_{ B' }(v) \geq \mathtt{x}_{ B' }(s)$. However, the fact that $\vab\paren*{ B } = \vab\paren*{ B' }$ implies that $v \leq s$. The monotonicity of the probability of sale (see \cref{lemma:typesoftypes:fg}) therefore implies that $\mathtt{x}_{ B' }(v) \leq \mathtt{x}_{ B' }(s)$, and the result follows.

\begin{figure}
\centering
\begin{tikzpicture}[domain=0:10,samples=1000]
    \draw[->] (-0.4,0) -- (10,0) node[below] {$v$};
    \draw[->] (0,-0.4) -- (0,4);
    \draw[dashed] (2,4) -- (2,-0.04) node[below] {$p$};
    \draw[dashed] (4.5,4) -- (4.5,-0.04) node[below] {$s$};
    \draw[black, ultra thick, domain=0:4.5] plot [smooth, tension=0.5] coordinates { (0,0) (0.5, 0.3) (1,0.5) (1.5, 0.8) (2,0.9) (2.5, 1) (3, 1.1) (3.5, 1.2) (4, 1.5) (4.5,1.578)};
    \draw[black, ultra thick, domain=4.5:10] plot (\x,{(4 - 51.24/( pow(0.1 + \x, 2)))});
    \draw[red, ultra thick, domain=0:10] plot (\x,{ and(\x > 1.7, \x <= 4.5) * (1.578) + (\x > 4.5) * (4 - 51.24/( pow(0.1 + \x, 2)))});
\end{tikzpicture}
\caption{Allocation rule of a general mechanism (black) and the corresponding bottom proper mechanism constructed in \cref{lemma:fixedprice-reduction} ({\color{red}{red}}). The allocation below $s$ is changed to a step function with the same area under the curve and the same probability of allocation at $s$.}
\label{fig:bottomproper}
\end{figure}
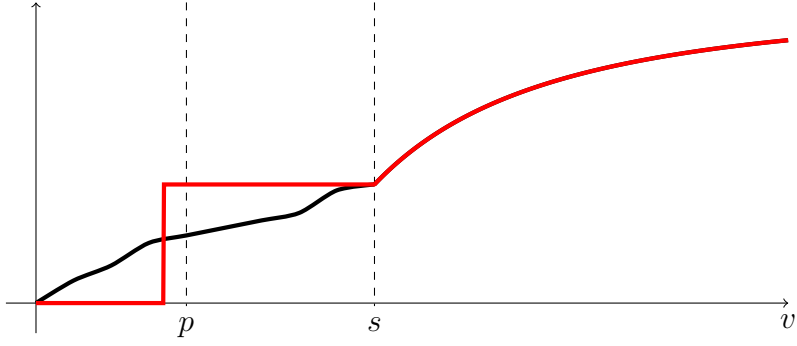

\medskip\noindent
\textbf{Step 3: $B'$ generates weakly higher revenue for Bob.} 
It remains to show that Bob's revenue when he uses $B'$ is at least that when he uses $B$. For this, note that $\abstart\paren*{ B } = \abstart\paren*{ B' }$, which means from \cref{lemma:oneseller} that we can simply compare the revenue of these mechanisms when Bob is the only seller and the distribution is modified to $D_s$. For this, we use \cref{lemma:nobottom}. To apply \cref{lemma:nobottom}, we have to show that \begin{inparaenum}[(1)] \item \label{item:lemmacondition:1} For all $0 \leq v \leq \hat{p}$, we have that $v \cdot f_s\paren*{ v } - \paren*{ 1 - F_s\paren*{ v } } \leq \hat{p} \cdot f_s\paren*{ \hat{p} } - \paren*{ 1 - F_s\paren*{ \hat{p} } }$. \item \label{item:lemmacondition:2} For all $\hat{p} \leq v < a + \paren*{ 1 - z } \cdot s$, we have that $\hat{p} \cdot f_s\paren*{ \hat{p} } - \paren*{ 1 - F_s\paren*{ \hat{p} } } \leq v \cdot f_s\paren*{ v } - \paren*{ 1 - F_s\paren*{ v } }$. \end{inparaenum} We show both of these together, by showing that for all $0 \leq v' \leq p$ and all $v' \leq v'' < a + \paren*{ 1 - z } \cdot s$, we have $v' \cdot f_s\paren*{ v' } - \paren*{ 1 - F_s\paren*{ v' } } \leq v'' \cdot f_s\paren*{ v'' } - \paren*{ 1 - F_s\paren*{ v'' } }$. This is straightforward if $f_s\paren*{ v' } \leq f_s\paren*{ v'' }$ so we assume otherwise. Let $w = v''$ if $v'' \leq p$ and $\frac{ v'' - a }{ 1 - z }$ otherwise, and note that $v' \leq w$. In case $D$ is regular, we have
\begin{align*}
v' \cdot f_s\paren*{ v' } - \paren*{ 1 - F_s\paren*{ v' } } &= v' \cdot f\paren*{ v' } - \paren*{ 1 - F\paren*{ v' } } + z \cdot \paren*{ 1 - F(s) } \tag{\cref{eq:fs,lemma:fs}} \\
&\leq f\paren*{ v' } \cdot \min\paren*{ 0, w - \frac{ 1 - F\paren*{ w } }{ f\paren*{ w } } } + z \cdot \paren*{ 1 - F(s) } \tag{As $D$ is regular and $v' \leq p < \mathtt{v}$} \\
&\leq f\paren*{ w } \cdot \min\paren*{ 0, w - \frac{ 1 - F\paren*{ w } }{ f\paren*{ w } } } + z \cdot \paren*{ 1 - F(s) } \tag{As $f\paren*{ v' } = f_s\paren*{ v' } > f_s\paren*{ v'' } \geq f\paren*{ w }$ by \cref{eq:fs}} \\
&\leq w \cdot f\paren*{ w } - \paren*{ 1 - F\paren*{ w } } + z \cdot \paren*{ 1 - F(s) } \\
&\leq v'' \cdot f_s\paren*{ v'' } - \paren*{ 1 - F_s\paren*{ v'' } } \tag{\cref{eq:fs,lemma:fs}} .
\end{align*}

On the other hand, if $D$ is DMR, we have:
\begin{align*}
v' \cdot f_s\paren*{ v' } - \paren*{ 1 - F_s\paren*{ v' } } &= v' \cdot f\paren*{ v' } - \paren*{ 1 - F\paren*{ v' } } + z \cdot \paren*{ 1 - F(s) } \tag{\cref{eq:fs,lemma:fs}} \\
&\leq w \cdot f\paren*{ w } - \paren*{ 1 - F\paren*{ w } } + z \cdot \paren*{ 1 - F(s) } \tag{As $D$ is DMR and $v' \leq w$} \\
&\leq v'' \cdot f_s\paren*{ v'' } - \paren*{ 1 - F_s\paren*{ v'' } } \tag{\cref{eq:fs,lemma:fs}} .
\end{align*}
\end{proof}


\subsection{Upper Bounds on the Best Stackelberg Equilibrium for Alice}
\label{sec:stackelberg.lowerbound}

Recall that \cref{thm:oneoptionalice-formal} implies that the best Stackelberg equilibirum for Alice gives her revenue at least $\frac{ \Gamma_D\paren*{ \mathtt{v} } }{4}$. We now show that there is a distribution $D$ for which the best Stackelberg equilibrium for Alice gives her revenue at most $\frac{ \Gamma_D\paren*{ \mathtt{v} } }{ \mathrm{e} }$. One such distribution $D$ is a point mass at $1$.
\begin{theorem}
\label{thm:tight2-onebob}
Let the distribution $D$ be a point mass at $1$. Then for any mechanism $A$ for Alice and any best response mechanism $B$ for Bob, it holds that $\rev_{ \mathsf{A}, \paren*{ A, B } }\paren*{ D } \leq \frac{1}{ \mathrm{e} }$.
\end{theorem}
A key lemma in the proof of \cref{thm:tight2-onebob} is a bound on the subgradients of (the lower convex envelope of) an arbitrary mechanism $A$ for Alice in terms of the revenue Bob can obtain from a fixed price mechanism. Formally, we have:
\begin{lemma}
\label{lemma:tight-onebob}
Let the distribution $D$ be a point mass at $1$ and $A$ be a mechanism for Alice. Let $R \geq \max_{ b \geq 0 } \rev_{ \mathsf{B}, \paren*{ A, \fp_b } }\paren*{ D }$ and $0 \leq z < 1 - R$ be arbitrary. Denoting by $\tilde{A}$ the lower convex envelope of $A$, we have for all subgradients $v_z$ of $\tilde{A}$ at $z$ that $v_z \leq \frac{ R }{ 1 - z }$. 
\end{lemma}
\begin{proof}
From the bounds on $z$, we have that $\frac{ R }{ 1 - z } < 1$. We proceed by contradiction. Suppose that we have a subgradient $\frac{ R }{ 1 - z } < v_z$. This means that there exists $\frac{ R }{ 1 - z } < v < \min\paren*{ 1, v_z }$. Consider what happens when Alice uses mechanism $A$ and Bob uses mechanism $\fp_v$. In this case, observe that the buyer goes to Alice first and purchases with probability $\mathtt{z}_A\paren*{ v }$, where $\mathtt{z}_A\paren*{ \cdot }$ is as defined in \cref{sec:prelim}. This means that $\rev_{ \mathsf{B}, \paren*{ A, \fp_v } }\paren*{ D } = v \cdot \paren*{ 1 - \mathtt{z}_A\paren*{ v } } \leq R$. It follows that $z < \mathtt{z}_A\paren*{ v }$. We get:
\begin{align*}
v \cdot \paren*{ \mathtt{z}_A\paren*{ v } - z } &< v_z \cdot \paren*{ \mathtt{z}_A\paren*{ v } - z } \tag{As $v < v_z$} \\
&\leq \tilde{A}\paren*{ \mathtt{z}_A\paren*{ v } } - \tilde{A}\paren*{ z } \tag{As $v_z$ is a subgradient of $\tilde{A}$ at $z$} \\
&\leq A\paren*{ \mathtt{z}_A\paren*{ v } } - \tilde{A}\paren*{ z } \\
&\leq v \cdot \paren*{ \mathtt{z}_A\paren*{ v } - z } ,
\end{align*}
a contradiction. Here, the last step is because the fact that $\tilde{A}$ is a lower convex envelope of $A$ implies that it pointwise at least any affine function that lower bounds $A$. The definition of $\mathtt{z}_A\paren*{ v }$ says that for all $z \in [0, 1]$, we have $A\paren*{ z } \geq A\paren*{ \mathtt{z}_A\paren*{ v } } + v \cdot \paren*{ z - \mathtt{z}_A\paren*{ v } }$ for all $z \in [0, 1]$, implying that the right hand side is one such affine function. 
\end{proof}

We are now ready to prove \cref{thm:tight2-onebob}
\begin{proof}[Proof of \cref{thm:tight2-onebob}]
Fix an arbitrary mechanism $A: [0, 1] \to \mathbb{R}$ for Alice and an arbitrary best response $B$ for Bob. As the valuation of the buyer is fixed to be $1$, we divide the proof into two cases based on whehter the buyer purchases from Alice first or from Bob first.
\begin{itemize}
\item \textbf{When the buyer buys from Alice first.} In this case, let $x_A$ and $x_B$ be the options chosen by the buyer from Alice and Bob respectively. Consider what happens if Bob deviates to a fixed price mechanism $\fp_q$ with $q = 1 - x_B + B\paren*{ x_B }$. This value of $q$ is designed so that the utility from Bob is preserved implying that, when the mechanism is $\fp_q$, the buyer still goes to Alice first and buys from Alice with the same probability. This means we must have $B\paren*{ x_B } \leq x_B = 1$ as otherwise, $B$ will not be a best response for Bob. 

Having shown that $B\paren*{ x_B } \leq x_B = 1$, observe that the buyer purchases from Alice with probability $\mathtt{z}_A\paren*{ B\paren*{ x_B } }$, where $\mathtt{z}_A\paren*{ \cdot }$ is as defined in \cref{sec:prelim}. This means that Alice's revenue is $A\paren*{ \mathtt{z}_A\paren*{ B\paren*{ x_B } } }$, which we claim equals $\tilde{A}\paren*{ \mathtt{z}_A\paren*{ B\paren*{ x_B } } }$, where $\tilde{A}$ denotes the lower convex envelope of $A$. To see why, note from the definition of $\mathtt{z}_A\paren*{ B\paren*{ x_B } }$ that $$A\paren*{ \mathtt{z}_A\paren*{ B\paren*{ x_B } } } + B\paren*{ x_B } \cdot \paren*{ z - \mathtt{z}_A\paren*{ B\paren*{ x_B } } } \leq A\paren*{ z } \text{ for all } z \in [0, 1].$$ As the left hand side is an affine function of $z$, it has to be lower than the lower convex envelope at every point $z \in [0, 1]$ and $A\paren*{ \mathtt{z}_A\paren*{ B\paren*{ x_B } } } = \tilde{A}\paren*{ \mathtt{z}_A\paren*{ B\paren*{ x_B } } }$ follows. 

Next, observe that $\tilde{A}$ is non-negative throughout and satisfies $\tilde{A}(0) = 0$. Also, recall that a convex function is differentiable except at countably many points and let $\tilde{A}'$ denote the derivative of $\tilde{A}$. From \cref{lemma:tight-onebob}, we have that Alice's revenue is:
\begin{multline*}
\tilde{A}\paren*{ \mathtt{z}_A\paren*{ B\paren*{ x_B } } } = \int_0^{ \mathtt{z}_A\paren*{ B\paren*{ x_B } } } \tilde{A}'\paren*{ z } \diff z \leq \int_0^{ \mathtt{z}_A\paren*{ B\paren*{ x_B } } } \frac{ B\paren*{ x_B } \cdot \paren*{ 1 - \mathtt{z}_A\paren*{ B\paren*{ x_B } } } }{ 1 - z } \diff z \\
\leq - \paren*{ 1 - \mathtt{z}_A\paren*{ B\paren*{ x_B } } } \cdot \ln\paren*{ 1 - \mathtt{z}_A\paren*{ B\paren*{ x_B } } } \leq \frac{1}{ \mathrm{e} } .
\end{multline*}

\item \textbf{When the buyer buys from Bob first.} As above, let $x_A$ and $x_B$ be the options chosen by the buyer from Alice and Bob respectively. We must have $B\paren*{ x_B } \leq 1$ for the buyer to get non-negative utility and can assume that $x_B < 1$ as otherwise the theorem is trivial. As above, let $\tilde{A}$ be the lower convex envelope of $A$. We first show that $A\paren*{ x_A } = \tilde{A}\paren*{ x_A }$. For this, note from our choice of $x_A$ and $x_B$ that for all $z \in [0, 1]$, we have:
\[
x_B - B\paren*{ x_B } + \paren*{ 1 - x_B } \cdot \paren*{ z - A\paren*{ z } } \leq x_B - B\paren*{ x_B } + \paren*{ 1 - x_B } \cdot \paren*{ x_A - A\paren*{ x_A } } .
\]
Using $x_B < 1$, this simplifies to $A\paren*{ z } \geq A\paren*{ x_A } + z - x_A$. As the right hand side is an affine function of $z$, we have by the definition of lower convex envelope that $\tilde{A}\paren*{ z } \geq A\paren*{ x_A } + z - x_A$ for all $z \in [0, 1]$. It follows that $\tilde{A}\paren*{ x_A } = A\paren*{ x_A }$. Again using our choice of $x_A$ and $x_B$, we get that the buyer will not strictly prefer to purchase from Alice first with these options, which is possible only if $x_A \cdot B\paren*{ x_B } \leq x_B \cdot A\paren*{ x_A }$. We get:
\begin{align*}
\rev_{ \mathsf{A}, \paren*{ A, B } }\paren*{ D } &= \paren*{ 1 - x_B } \cdot A\paren*{ x_A } \\
&= \tilde{A}\paren*{ x_A } - x_B \cdot A\paren*{ x_A } \\
&\leq \tilde{A}\paren*{ x_A } - x_A \cdot B\paren*{ x_B } \tag{As $x_A \cdot B\paren*{ x_B } \leq x_B \cdot A\paren*{ x_A }$} \\
&\leq \tilde{A}\paren*{ 1 - B\paren*{ x_B } } - \paren*{ 1 - B\paren*{ x_B } } \cdot \paren*{ 1 - x_A } \tag{As $\tilde{A}\paren*{ z } \geq \tilde{A}\paren*{ x_A } + z - x_A$ for all $z \in [0, 1]$ and $\tilde{A}\paren*{ x_A } = A\paren*{ x_A }$} \\
&\leq \tilde{A}\paren*{ 1 - B\paren*{ x_B } } \tag{As $x_A, B\paren*{ x_B } \leq 1$} .
\end{align*}
Observing that $\tilde{A}(0) = 0$ and letting $\tilde{A}'$ denote the derivative of $\tilde{A}$, we have from \cref{lemma:tight-onebob} that:
\[
\rev_{ \mathsf{A}, \paren*{ A, B } }\paren*{ D } \leq \int_0^{ 1 - B\paren*{ x_B } } \tilde{A}'\paren*{ z } \diff z \leq \int_0^{ 1 - B\paren*{ x_B } } \frac{ B\paren*{ x_B } }{ 1 - z } \diff z = - B\paren*{ x_B } \cdot \ln\paren*{ B\paren*{ x_B } } \leq \frac{1}{ \mathrm{e} } .
\]
\end{itemize}
\end{proof}

\subsection{When Bob Responds with a Fixed Price Mechanism}
\label{sec:oneoptionbob}

To finish this section we show that, regardless of the distribution $D$, if Bob is restricted to respond with a fixed price mechanism, the bound of $\frac{1}{ \mathrm{e} }$ in \cref{thm:tight2-onebob} is tight.  That is, there is a mechanism for Alice such that, if Bob chooses his revenue-maximizing fixed-price mechanism in response, then Alice's resulting revenue is within a factor of $1/e$ of the optimal monopolist's revenue.

\begin{theorem}\label{thm:alicerev-onebob}
    Fix any value distribution $D$ and let $\mathtt{v} = \argmax_v \Gamma_D(v)$ be the corresponding revenue-maximizing price.  Then there exists a mechanism $A$ for Alice such that, for any fixed-price mechanism $\fp_q$ for Bob with $q \in \argmax_q\{ \rev_{ \mathsf{B}, \paren*{A, \fp_q}}\paren*{ D } \}$, we have $\rev_{ \mathsf{A}, \paren*{ A, \fp_q } }\paren*{ D } \geq \Gamma_D(\mathtt{v})/e$.
\end{theorem}



For this setting, as before, the buyer can either go to Alice first or go to Bob first. However, as Bob is selling using a fixed price mechanism, whenever a buyer purchases from Bob, he is sure to get the item and will not go to Alice. Thus, without loss of generality, we can assume that the buyer goes to Alice first, picks the option corresponding to $x$ from Alice, for some $x \in [0, 1]$, and if he is not allocated the item (this happens with probability $1 - x$), may or may not purchase from Bob depending on whether or not it improves his utility. 

Given this, our idea for Alice is to select a mechanism against which Bob will choose to act like a monopolist and set the Myerson reserve price $\texttt{v} = \argmax_v \Gamma_D(v)$.  To do so, Alice constructs a mechanism that makes Bob indifferent between setting price $\texttt{v}$ and any (sufficiently close) lower price.  Writing $\mathcal{M} = \Gamma_D(\texttt{v})$ for the optimal revenue, we show that this is accomplished by the following mechanism for Alice: 
\begin{equation}
\label{eq:alice-onebob-main}
A\paren*{ x } = \begin{cases}
\int_0^x \Gamma_D^{ -1 }\paren*{ \frac{ \mathcal{M} }{ \mathrm{e} \cdot \paren*{ 1 - z } } } \diff z , &\text{~if~} x \leq 1 - \frac{1}{ \mathrm{e} } \\
\infty , &\text{~otherwise~} 
\end{cases} .
\end{equation}
Here, $\mathrm{e}$ is the base of the natural logarithm and $\infty$ represents a value too large for the buyer to pay. Note that $x \leq 1 - \frac{1}{ \mathrm{e} }$ implies that the argument $\Gamma_D^{ -1 }$ lies in its domain and thus $A\paren*{ x }$ is well defined. 

In Lemma~\ref{lemma:bobrev-onebob}, we show that against this mechanism, Bob is indifferent between setting any price between $\Gamma^{-1}(\mathcal{M}/e)$ and $\Gamma^{-1}(\mathcal{M}) = \texttt{v}$, and obtains strictly lower revenue outside that range.  In other words, Bob has no incentive to undercut the monopolist price in order to try and attract more buyer types away from Alice.

The final step is to calculate the revenue of this mechanism $A$ when Bob uses posted price $\texttt{v}$.  To do this, we first characterize the allocation rule generated by Alice's menu in Lemma~\ref{lemma:za-onebob}, then in Lemma~\ref{lemma:alicerev-onebob} we directly calculate revenue given Bob's choice of mechanism to conclude that $\rev_{ \mathsf{A}, \paren*{ A, \fp_b } }\paren*{ D } = \frac{ \mathcal{M} }{ \mathrm{e} }$, completing the proof of Theorem~\ref{thm:alicerev-onebob}.

\section{Nash Equilibrium with Fixed Buyer Type}
\label{sec:fixedtype}

In this section, we consider pure Nash equilibria and prove \cref{thm:nonash-informal}. As before, we fix a distribution $D$ from which the buyer's valuations are drawn.  For any such distribution $D$, a pure Nash equilibrium always exists for Alice and Bob: they could each choose a fixed price mechanism with price $0$, in which case neither seller has a deviation that would generate positive revenue.  We show that this is essentially the only Nash equilibrium.

\begin{theorem}
\label{thm:nonash}
Let the distribution $D$ be a point mass at $1$. Suppose that $A$ and $B$ are mechanisms for Alice and Bob respectively such that:
\begin{align*}
\rev_{ \mathsf{A}, \paren*{ A, B } }\paren*{ D } &= \max_{ A' } \rev_{ \mathsf{A}, \paren*{ A', B } }\paren*{ D } , \\
\rev_{ \mathsf{B}, \paren*{ A, B } }\paren*{ D } &= \max_{ B' } \rev_{ \mathsf{B}, \paren*{ A, B' } }\paren*{ D } .
\end{align*}
Then, it holds that $\rev_{ \mathsf{A}, \paren*{ A, B } }\paren*{ D } = \rev_{ \mathsf{B}, \paren*{ A, B } }\paren*{ D } = 0$.
\end{theorem}
\begin{proof}
Fix $A$ and $B$ as in the theorem statement and suppose for the sake of contradiction that at least one of $\rev_{ \mathsf{A}, \paren*{ A, B } }\paren*{ D }$ and $\rev_{ \mathsf{B}, \paren*{ A, B } }\paren*{ D }$ is positive. Facing these mechanisms, the buyer (with valuation $1$) can either go to Alice first or to Bob first. If the buyer goes to Alice first and picks the options corresponding to $x_A$ and $x_B$ from Alice and Bob respectively, he will get a utility of:
\[
\utilab\paren*{ A, B, x_A, x_B } = x_A - A\paren*{ x_A } + \paren*{ 1 - x_A } \cdot \paren*{ x_B - B\paren*{ x_B } } .
\]
Similarly, if the buyer goes to Bob first and picks the options corresponding to $x_A$ and $x_B$ from Alice and Bob respectively, he will get a utility of:
\[
\utilba\paren*{ A, B, x_A, x_B } = x_B - B\paren*{ x_B } + \paren*{ 1 - x_B } \cdot \paren*{ x_A - A\paren*{ x_A } } .
\]
Of course, the buyer will pick whichever option that maximizes his utility. 
Assume without loss of generality that the buyer goes to Alice first. Let $x_A$ and $x_B$ be the options chosen from Alice's and Bob's menu respectively. Thus, we have that $\rev_{ \mathsf{A}, \paren*{ A, B } }\paren*{ D } = A\paren*{ x_A }$ and $\rev_{ \mathsf{B}, \paren*{ A, B } }\paren*{ D } = \paren*{ 1 - x_A } \cdot B\paren*{ x_B }$. We first claim that since one of the sellers is getting positive expected revenue, both must be getting positive expected revenue.
\begin{claim}
\label{claim:nonash-1}
We have $\rev_{ \mathsf{A}, \paren*{ A, B } }\paren*{ D }, \rev_{ \mathsf{B}, \paren*{ A, B } }\paren*{ D } > 0$.
\end{claim}
\begin{proof}
We show the former as the proof for the latter is analogous. Suppose for the sake of contradiction that $\rev_{ \mathsf{A}, \paren*{ A, B } }\paren*{ D } = \max_{ A' } \rev_{ \mathsf{A}, \paren*{ A', B } }\paren*{ D } = 0$ and for all $\delta > 0$, consider the mechanism $A'_{ \delta }$ for Alice where the only lottery offered to the buyer is to purchase the item with probability $1$ at a price of $\delta$. Due to the foregoing equation, we have that, when Alice uses the mechanism $A'_{ \delta }$, the buyer must go to Bob first and pick an option corresponding to probability $1$ from him before going to Alice. This means that:
\[
\utilab\paren*{ A'_{ \delta }, B, 1, x_B } = 1 - \delta \leq 1 - B\paren*{ 1 } = \utilba\paren*{ A'_{ \delta }, B, 1, 1 } . 
\]
As $\delta > 0$ was arbitrary, we get that $B\paren*{ 1 } = 0$. It follows from the choice of $x_A$ and $x_B$ that $\utilab\paren*{ A, B, x_A, x_B } \geq \utilba\paren*{ A, B, x_A, 1 } = 1$. However, this means that:
\[
1 \leq x_A - A\paren*{ x_A } + \paren*{ 1 - x_A } \cdot \paren*{ x_B - B\paren*{ x_B } } \leq x_A + \paren*{ 1 - x_A } \cdot x_B \leq 1 .
\]
Thus, there must be equality throughout implying that $\rev_{ \mathsf{A}, \paren*{ A, B } }\paren*{ D } = \rev_{ \mathsf{B}, \paren*{ A, B } }\paren*{ D } = 0$, a contradiction.
\end{proof}

Next we claim that the chosen mechanisms must be such that the buyer is choosing, from each menu, whichever lottery would maximize their utility \emph{not accounting for the presence of the other seller}.  For Bob this is an immediate implication of our assumption that the buyer visits Bob second.  For Alice, this occurs because otherwise there is slack in Alice's choice of mechanism that could be leveraged for increased revenue for either Alice or Bob. Formally, we show that:

\begin{claim}
\label{claim:nonash-2}
We have $0 < A\paren*{ x_A } < x_A < 1$ and $0 < B\paren*{ x_B } < x_B$. We also have:
\[
x_A - A\paren*{ x_A } = \max_{ z \in [0, 1] } \paren*{ z - A\paren*{ z } } \hspace{1cm}\text{and}\hspace{1cm} x_B - B\paren*{ x_B } = \max_{ z \in [0, 1] } \paren*{ z - B\paren*{ z } } .
\]
\end{claim}

Assuming \cref{claim:nonash-2} for now, we are ready to complete the proof of \cref{thm:nonash}. Let $\delta = \frac{ x_B - B\paren*{ x_B } }{ 2 }$ and note from \cref{claim:nonash-2} that $\delta > 0$ and $B\paren*{ x_B } + \delta < x_B$. Consider the mechanism $B'$ for Bob where the only lottery offered to the buyer is to purchase the item with probability $x_B$ at a price of $B\paren*{ x_B } + \delta$. If the buyer goes to Bob first when faced with the mechanisms $A$ and $B'$, Bob gets a revenue of $B\paren*{ x_B } + \delta > B\paren*{ x_B } \geq \rev_{ \mathsf{B}, \paren*{ A, B } }\paren*{ D }$, a contradiction. Similarly, if the buyer goes to Alice first and picks an option $x'_A \leq x_A$, Bob gets a revenue of at least $\paren*{ 1 - x_A } \cdot \paren*{ B\paren*{ x_B } + \delta } > \rev_{ \mathsf{B}, \paren*{ A, B } }\paren*{ D }$ as $x_A < 1$ from \cref{claim:nonash-2}, a contradiction. Thus, the buyer must go to Alice first and pick some option, say $x'_A > x_A$, before going to Bob and picking the option $x_B$. This means that $\utilab\paren*{ A, B', x_A, x_B } \leq \utilab\paren*{ A, B', x'_A, x_B }$ which gives:
\[
x_A - A\paren*{ x_A } + \paren*{ 1 - x_A } \cdot \paren*{ x_B - B'\paren*{ x_B } } \leq x'_A - A\paren*{ x'_A } + \paren*{ 1 - x'_A } \cdot \paren*{ x_B - B'\paren*{ x_B } } .
\]
Using \cref{claim:nonash-2}, we have $x'_A - A\paren*{ x'_A } \leq x_A - A\paren*{ x_A }$. Plugging in, we get $\paren*{ x'_A - x_A } \cdot \paren*{ x_B - B'\paren*{ x_B } } \leq 0$. As $B'\paren*{ x_B } = B\paren*{ x_B } + \delta < x_B$ and $x'_A > x_A$, we have a contradiction.

\end{proof}

\section{Conclusions and Open Problems}
\label{sec:conclusion}

In this paper we studied a game played between two competing mechanism designers with the ability to randomize allocations.  When competition is perfect and each seller is selling an identical good, there are instances where only the trivial pure Nash equilibrium exists.  However, in the Stackelberg setting where one seller can commit to her choice of mechanism before the other responds, it is possible to approximate the revenue raised by a monopolist, for any buyer value distribution that satisfies either the regularity condition or the decreasing marginal revenue condition.

Many questions and future directions are left open.  On a technical level, there is a gap left between our upper and lower bounds on the leader's Stackelberg payoff; the approximation factor versus the monopolist revenue lies in $\bracket*{ \mathrm{e}, 4 }$.  It appears that resolving this gap will require the use of multi-lottery mechanisms: we show that the factor $4$ is tight if Alice is restricted to using a single-lottery mechanism, and the factor $\mathrm{e}$ is tight if Bob is required to use a posted-price mechanism.  Beyond resolving the exact approximation factor, it would be interesting to resolve the structure of the exact optimal mechanisms chosen at Stackelberg equilibrium, and/or derive conditions under which simple mechanism structures are optimal.

There are also numerous extensions to our model that would be interesting to explore.  Sellers could have multiple items to sell to a buyer with combinatorial preferences, the buyers may not be risk neutral, or there could be multiple buyers interacting with the competing sellers.  The latter could involve either synchronous or asynchronous timing of interactions.  Also related to timing, we assume in our model that the buyer visits the mechanisms in sequence, resolving one before choosing whether to visit the other.  One could alternatively consider a setting where a buyer can participate in multiple mechanisms simultaneously, partially resolving one and then switching to another.  In such an environment, could sellers benefit by designing multi-round protocols that (at equilibrium) account for the possibility of a buyer leaving partway for a competitor and returning later?

\bibliographystyle{alpha}
\bibliography{refs}

\appendix


\section{Missing Proofs: Taxation Principle}
\label{sec:prelim-app}

In this section we provide a proof of \Cref{obs:taxation}.  Recall the statement:

\medskip

\noindent
\textbf{Observation 2.1} (Taxation Principle).  \emph{For any mechanism $\mathcal{M}$ for either seller there is a strategically equivalent proper pricing mechanism. In the corresponding pricing rule, $M(x)$ is the minimum expected payment over any distribution of buyer messages that results in an allocation with probability at least $x$ (or $M(x) = \infty$ if no such message exists).}

\medskip

To see why \cref{obs:taxation} holds intuitively, note that it simply replaces each mechanism with an equivalent pricing rule assuming the mechanism was run by a monopolist.  What \cref{obs:taxation} claims is that, in our duopolist setting, the original joint mechanism is also equivalent to the resulting joint pricing rules. The reason is that since the buyer participates in any mechanism to completion before engaging with the other mechanism, only the outcome of each mechanism ({\em i.e.}, its distribution over allocation and payment) can impact buyer behavior in the other.  
Moreover, as the buyer is risk-neutral with quasi-linear utility, only the resulting probability of allocation and \emph{expected} payment can influence buyer behavior. By definition, the constructed pricing rules do not change allocation probabilities or expected payments, so \cref{obs:taxation} follows.

More formally, suppose the two sellers select mechanisms $\mathcal{M}_1, \mathcal{M}_2$.  For any undominated strategy $(i,\sigma_i,\sigma_{-i})$ for the buyer, let $a_i$ and $t_i$ denote the resulting allocation probability and expected payment from mechanism $\mathcal{M}_i$, and let $a_{-i}$ and $t_{-i}$ denote the allocation probability and expected payment from $\mathcal{M}_{-i}$ conditional on not obtaining the item from mechanism $i$.  Then the buyer's expected utility from this strategy is
\[ v a_i - t_i + (1-a_i)( v a_{-i} - t_{-i} ). \]
In particular, the buyer's expected utility is entirely determined by $(a_i, a_{-i}, t_i, t_{-i})$.

Let $M_1$, $M_2$ be the corresponding proper pricing rules described in \cref{obs:taxation} for mechanisms $\mathcal{M}_1, \mathcal{M}_2$.  Consider what would occur if seller $1$ deviated from mechanism $\mathcal{M}_1$ to the pricing mechanism $M_1$.
By construction of $M_1$, any profile of allocation probabilities and expected transfers $(a_i, a_{-i}, t_i, t_{-i})$ is attainable by the buyer from the profile of mechanisms $(M_1, \mathcal{M}_2)$ if and only if they are obtainable via an undominated strategy for the mechanisms $(\mathcal{M}_1, \mathcal{M}_2)$.  Thus the corresponding obtainable expected utilities for the buyer are identical as well.  The correspondence of buyer-optimal behavior thus results in the same buyer utility and transfers to the sellers, and in particular for seller $1$, so it is without loss for seller $1$ to employ pricing mechanism $M_1$ instead of $\mathcal{M}_1$.  By symmetry, the same is true for seller $2$.  Thus each seller has a strategically equivalent proper pricing mechanism, regardless of the choice of the other seller, as claimed.


\section{Missing Proofs from the Stackelberg Setting}
\label{app:stackelberg}

In this appendix, we provide the proofs omitted from the analysis described in \cref{sec:stackelberg}. 

\subsection{Lemmas about Buyer Behavior}

Throughout this section, we fix a single lottery mechanism $A$ for Alice with probability $0 < z < 1$ and price $p > 0$ and let $a = pz$. We start with the following simple lemma showing that a buyer with a large type will choose a large value of $x$ from Bob.
\begin{lemma}
\label{lemma:util}
Let $B$ be a mechanism for Bob and $\epsilon > 0$. For all $v > \max\paren*{ p, \frac{ B(1) }{ \epsilon \cdot \paren*{ 1 - z } } }$ and $x \leq 1 - \epsilon$, we have:
\begin{align*}
\utilba\paren*{ B, x, v } &< \utilba\paren*{ B, 1, v } \\
\utilab\paren*{ B, x, v } &< \utilab\paren*{ B, 1, v } .
\end{align*}
\end{lemma}
\begin{proof}

 For the first one, note using \cref{eq:utilba} that:
\begin{align*}
\utilba\paren*{ B, x, v } - \utilba\paren*{ B, 1, v } &\leq B\paren*{ 1 } - B\paren*{ x } + \paren*{ 1 - x } \cdot \paren*{ z v - a - v } \\
&\leq B\paren*{ 1 } - v \cdot \paren*{ 1 - x } \cdot \paren*{ 1 - z } \\
&< 0 . \tag{As $x \leq 1 - \epsilon$ and $v > \frac{ B(1) }{ \epsilon \cdot \paren*{ 1 - z } }$}
\end{align*}
Similarly, we have using \cref{eq:utilab} that:
\begin{align*}
\utilab\paren*{ B, x, v } - \utilab\paren*{ B, 1, v } &\leq \paren*{ 1 - z } \cdot \paren*{ x v - B\paren*{ x } - v + B\paren*{ 1 } } \\
&\leq \paren*{ 1 - z } \cdot \paren*{ B\paren*{ 1 } - v \cdot \paren*{ 1 - x } } \\
&< 0 . \tag{As $x \leq 1 - \epsilon$ and $v > \frac{ B(1) }{ \epsilon \cdot \paren*{ 1 - z } }$}
\end{align*}
\end{proof}

We now establish some properties about the buyer behavior in our duopolist setting that will be useful throughout \cref{sec:stackelberg}. First, we show that a buyer only buys from Alice first if his value exceeds the bang-per-buck price $p$ offered by Alice. 
\begin{lemma}
\label{lemma:typesoftypes:abc}
For all mechanisms $B$ and all $v \in \vab\paren*{ B }$, we have $\mathtt{x}_B(v) > 0$ and $v > p$. It follows that $p \leq \abstart\paren*{ B }$.
\end{lemma}
\begin{proof}
Note that $v \cdot \mathtt{x}_B(v) > B\paren*{ \mathtt{x}_B(v) }$ because if not, the buyer can get the same utility by purchasing from Bob first, and our tie-breaking would contradict $v \in \vab\paren*{ B }$. It follows that $\mathtt{x}_B(v) > 0$. To show that $v > p$, note that $v \in \vab\paren*{ B }$ and our tie-breaking implies that $\utilba\paren*{ B, \mathtt{x}_B(v), v } < \utilab\paren*{ B, \mathtt{x}_B(v), v }$. Simplifying using \cref{eq:utilba,eq:utilab}, we get:
\[
z \cdot \paren*{ v \cdot \mathtt{x}_B(v) - B\paren*{ \mathtt{x}_B(v) } } < \mathtt{x}_B(v) \cdot \max\paren*{ 0, z v - a } .
\]
As the left side is positive, we get $z v > a$ and we are done. 
\end{proof}

Next, we show that the option $\mathtt{x}_B\paren*{ v }$ chosen by the buyer is non-decreasing in $v$.
\begin{lemma}
\label{lemma:typesoftypes:fg}
For all mechanisms $B$ and all $0 \leq v \leq v'$, we have $\mathtt{x}_B\paren*{ v } \leq \mathtt{x}_B\paren*{ v' }$. 
\end{lemma}
\begin{proof}
To show that $\mathtt{x}_B\paren*{ v }$ is non-decreasing, it suffices to show that it is non-decreasing in the interval $[0, p]$ and also in the interval $\co*{ p, \infty }$. We prove both these parts separately. For the first part, note that $v \leq v' \in [0, p]$ implies from \cref{lemma:typesoftypes:abc} that $v, v' \in \vba\paren*{ B }$. This means that $\utilba\paren*{ B, \mathtt{x}_B\paren*{ v' }, v } \leq \utilba\paren*{ B, \mathtt{x}_B\paren*{ v }, v }$ and $\utilba\paren*{ B, \mathtt{x}_B\paren*{ v }, v' } \leq \utilba\paren*{ B, \mathtt{x}_B\paren*{ v' }, v' }$. Adding and using \cref{eq:utilba}, we get:
\[
\paren*{ \mathtt{x}_B\paren*{ v' } - \mathtt{x}_B\paren*{ v } } \cdot \paren*{ v' - v } \geq 0 .
\]
The result follows. For the second part, for any $v \leq v' \in \co*{ p, \infty }$, consider the case $v \in \vba\paren*{ B }$ and $v' \in \vab\paren*{ B }$ (the other cases are similar). This means that $\utilab\paren*{ B, \mathtt{x}_B\paren*{ v' }, v } \leq \utilba\paren*{ B, \mathtt{x}_B\paren*{ v }, v }$ and $\utilba\paren*{ B, \mathtt{x}_B\paren*{ v }, v' } \leq \utilab\paren*{ B, \mathtt{x}_B\paren*{ v' }, v' }$. Adding and using \cref{eq:utilba,eq:utilab}, we get:
\[
\paren*{ 1 - z } \cdot \paren*{ \mathtt{x}_B\paren*{ v' } - \mathtt{x}_B\paren*{ v } } \cdot \paren*{ v' - v } \geq 0 .
\]
As $z < 1$, the result follows.
\end{proof}

The next lemma shows that the bang-per-buck price paid by the buyer to Bob is at most his valuation, and is at most $p$ if and only the buyer goes to Bob first. 
\begin{lemma}
\label{lemma:typesoftypes:d}
Let $B$ be a mechanism and $v \geq 0$. If we have $v \in \vba\paren*{ B }$ and $\mathtt{x}_B(v) > 0$, then $\frac{ B\paren*{ \mathtt{x}_B(v) } }{ \mathtt{x}_B(v) } \leq \min\paren*{ p, v }$. Else, if $v \in \vab\paren*{ B }$, we have $p < \frac{ B\paren*{ \mathtt{x}_B(v) } }{ \mathtt{x}_B(v) } < v$.
\end{lemma}
\begin{proof}
For all $v > p$ such that $\mathtt{x}_B(v) > 0$, note by our tiebreaking that we have $v \in \vba\paren*{ B }$ if and only if $\utilab\paren*{ B, \mathtt{x}_B(v), v } \leq \utilba\paren*{ B, \mathtt{x}_B(v), v }$. Using \cref{eq:utilba,eq:utilab}, we get that $v \in \vba\paren*{ B }$ if and only if $\frac{ B\paren*{ \mathtt{x}_B(v) } }{ \mathtt{x}_B(v) } \leq p$. Next, note that for all $v \leq p$ such $\mathtt{x}_B(v) > 0$, we have $v \in \vba\paren*{ B }$ by \cref{lemma:typesoftypes:abc} which means that $\utilba\paren*{ B, 0, v } \leq \utilba\paren*{ B, \mathtt{x}_B(v), v }$. Using \cref{eq:utilba}, we get that $\frac{ B\paren*{ \mathtt{x}_B(v) } }{ \mathtt{x}_B(v) } \leq v$. Finally, for all $v \in \vab\paren*{ B }$ which implies $v > p$ and $\mathtt{x}_B(v) > 0$ by \cref{lemma:typesoftypes:abc}, we have by our tie breaking that $\utilab\paren*{ B, 0, v } < \utilab\paren*{ B, \mathtt{x}_B(v), v }$ which from \cref{eq:utilba,eq:utilab} gives $\frac{ B\paren*{ \mathtt{x}_B(v) } }{ \mathtt{x}_B(v) } < v$ as $z < 1$. Combining gives the lemma.
\end{proof}

Now, we show that the bang-per-buck price paid by the buyer is also non-decreasing in $v$, by first showing a more general inequality.
\begin{lemma}
\label{lemma:typesoftypes:e}
Let $B$ be a mechanism and $v \geq 0$ satisfy $\mathtt{x}_B(v) > 0$. For all $0 < \delta < v$ and $\mathtt{x}_B(v) - \frac{ \delta \cdot \mathtt{x}_B(v) }{ 2v } < x' \leq 1$, we have $\frac{ B\paren*{ x' } }{ x' } \geq \frac{ B\paren*{ \mathtt{x}_B(v) } }{ \mathtt{x}_B(v) } - \delta$. Combining with \cref{lemma:typesoftypes:fg}, for such $v$, we have for all $v' \geq v$ that $\frac{ B\paren*{ \mathtt{x}_B(v) } }{ \mathtt{x}_B(v) } \leq \frac{ B\paren*{ \mathtt{x}_B\paren*{ v' } } }{ \mathtt{x}_B\paren*{ v' } }$.
\end{lemma}
\begin{proof}
Observe that $\mathtt{x}_B(v) > 0$ implies that $x' > 0$. By setting $v' = v$ if $v \in \vab\paren*{ B }$ or if $v \leq p$ and $v' = a + \paren*{ 1 - z } \cdot v$ otherwise, note that we have from \cref{eq:utilba,eq:utilab} that:
\[
x' \cdot \paren*{ v' - \frac{ B\paren*{ x' } }{ x' } } \leq \mathtt{x}_B(v) \cdot \paren*{ v' - \frac{ B\paren*{ \mathtt{x}_B(v) } }{ \mathtt{x}_B(v) } } .
\]
Note from the definition of $v'$ and \cref{lemma:typesoftypes:d} that the last factor is non-negative. Using the bounds $v' \leq v$ and $\mathtt{x}_B(v) \leq x' + \frac{ \delta \cdot \mathtt{x}_B(v) }{ 2v }$, this rearranges to:
\[
x' \cdot \paren*{ v' - \frac{ B\paren*{ x' } }{ x' } } \leq x' \cdot \paren*{ v' - \frac{ B\paren*{ \mathtt{x}_B(v) } }{ \mathtt{x}_B(v) } } + \frac{ \delta \cdot \mathtt{x}_B(v) }{ 2 } .
\]
As $\frac{ \mathtt{x}_B(v) }{ 2 } \leq x'$, we can divide by $x'$ to get the lemma.
\end{proof}

Now, we show that a buyer with type $\abstart\paren*{ B }$ always buys from Bob first, unless $\abstart\paren*{ B } = \infty$.
\begin{lemma}
\label{lemma:typesoftypes:h}
For all mechanisms $B$, if $\vab\paren*{ B } \neq \emptyset$ (that is, if $\abstart\paren*{ B } < \infty$), then, we have $\abstart\paren*{ B } \in \vba\paren*{ B }$. 
\end{lemma}
\begin{proof}
The high-level reason this is true is that we tiebreak in favor of Bob. Formally, let $s = \abstart\paren*{ B }$ and observe that we can assume $s > p$ as otherwise we are done by \cref{lemma:typesoftypes:abc}. Also, note that $s - \delta \in \vba\paren*{ B }$ for all $\delta > 0$. This means that, for all $\delta > 0$, there exists $x_{ \delta } \in [0, 1]$ such that $\utilab\paren*{ B, \mathtt{x}_B(s), s - \delta } \leq \utilba\paren*{ B, \mathtt{x}_B\paren*{ s - \delta }, s - \delta }$. Using \cref{eq:utilba,eq:utilab}, we get:
\[
\paren*{ 1 - z } \cdot \paren*{ \mathtt{x}_B(s) \cdot s - B\paren*{ x } } - \delta \leq \mathtt{x}_B\paren*{ s - \delta } \cdot \paren*{ a + \paren*{ 1 - z } \cdot s } - B\paren*{ \mathtt{x}_B\paren*{ s - \delta } } .
\]
Define $s' = a + \paren*{ 1 - z } \cdot s$ and observe that the right hand side is upper bounded by $s' \cdot \mathtt{z}_B\paren*{ s' } - B\paren*{ \mathtt{z}_B\paren*{ s' } }$. As this bound holds for all $\delta > 0$, we have $\paren*{ 1 - z } \cdot \paren*{ \mathtt{x}_B(s) \cdot s - B\paren*{ x } } \leq s' \cdot \mathtt{z}_B\paren*{ s' } - B\paren*{ \mathtt{z}_B\paren*{ s' } }$. Adding $z s - a$ on both sides and \cref{eq:utilba,eq:utilab}, we get $\utilab\paren*{ B, \mathtt{x}_B(s), s } \leq \utilba\paren*{ B, \mathtt{z}_B\paren*{ s' }, s }$. This is impossible if $s \in \vab\paren*{ B }$, and we are done.
\end{proof}

Finally, we show a lemma to handle the corner case mentioned in \cref{lemma:fixedprice-reduction}.
\begin{lemma}
\label{lemma:typesoftypes:i}
Let $B$ be a mechanism for Bob with $\vab\paren*{ B } = \emptyset$. Define the mechanism $B'$ as:
\[
B'(x) = \begin{cases}
B(x) , &\text{~if~} x < 1 \\
\sup_{ v : \mathtt{x}_B\paren*{ v } > 0 } \frac{ B\paren*{ \mathtt{x}_B\paren*{ v } } }{ \mathtt{x}_B\paren*{ v } } , &\text{~if~} x = 1
\end{cases}
\]
Then, we have $\vab\paren*{ B' } = \emptyset$ and $B'(1) = \sup_{ v : \mathtt{x}_{ B' }\paren*{ v } > 0 } \frac{ B'\paren*{ \mathtt{x}_{ B' }\paren*{ v } } }{ \mathtt{x}_{ B' }\paren*{ v } }$ and $\rev_{ \mathsf{B}, \paren*{ A, B } }\paren*{ D } \leq \rev_{ \mathsf{B}, \paren*{ A, B' } }\paren*{ D }$.
\end{lemma}
\begin{proof}
We start with the following helper claim.
\begin{claim}
\label{claim:typesoftypes:ihelper}
Let $B''$ be a mechanism for Bob with $\vab\paren*{ B'' } = \emptyset$ and let $\nabla'' = \sup_{ v : \mathtt{x}_{ B'' }\paren*{ v } > 0 } \frac{ B''\paren*{ \mathtt{x}_{ B'' }\paren*{ v } } }{ \mathtt{x}_{ B'' }\paren*{ v } }$. Then, $\nabla'' \leq p$ and the supremum in $\nabla''$ is attained at $v$ if and only if $\mathtt{x}_{ B'' }\paren*{ v } = 1$.
\end{claim}
\begin{proof}
That we have $\nabla'' \leq p$ is from \cref{lemma:typesoftypes:d}. Observe that the ``if'' direction follows easily from \cref{lemma:typesoftypes:e}. For the ``only if'' direction, assume for contradiction that the supremum is attained at some point $v$ for which $\mathtt{x}_{ B'' }\paren*{ v } < 1$. Then, we have from \cref{lemma:typesoftypes:e} that it also attained at $v'$ for all $v' \geq v$. Moreover, for a large enough $v'$, \cref{lemma:util} says that $\mathtt{x}_{ B'' }\paren*{ v } < \mathtt{x}_{ B'' }\paren*{ v' }$. By our tie-breaking, we must have that $\utilba\paren*{ B'', \mathtt{x}_{ B'' }\paren*{ v' }, v } < \utilba\paren*{ B'', \mathtt{x}_{ B'' }\paren*{ v }, v }$. Using \cref{eq:utilba}, and the fact that the supremum is attained at both $v$ and $v'$, this gives a contradiction. 
\end{proof}

Let $\nabla = \sup_{ v : \mathtt{x}_B\paren*{ v } > 0 } \frac{ B\paren*{ \mathtt{x}_B\paren*{ v } } }{ \mathtt{x}_B\paren*{ v } }$ and $\nabla' = \sup_{ v : \mathtt{x}_{ B' }\paren*{ v } > 0 } \frac{ B'\paren*{ \mathtt{x}_{ B' }\paren*{ v } } }{ \mathtt{x}_{ B' }\paren*{ v } }$ for convenience. If the supremum in $\nabla$ is attained we have $B = B'$ from \cref{claim:typesoftypes:ihelper} and there is nothing to show. Thus, we assume throughout that the supremum in $\nabla$ is never attained. 

For this case, \cref{claim:typesoftypes:ihelper} says that $\nabla \leq p$ and $\mathtt{x}_B\paren*{ v } < 1$ for all $v \geq 0$. As the only option that can possibly change between $B$ and $B'$ is the option corresponding to $x = 1$, it follows that for all $v \geq 0$, we have $\mathtt{x}_{ B' }\paren*{ v } \in \set*{ \mathtt{x}_B\paren*{ v }, 1 }$. Again using the fact that $\mathtt{x}_B\paren*{ v } < 1$ and $B'$ matches $B$ at every point but $x = 1$, it follows that $\mathtt{x}_{ B' }\paren*{ v } > 0$ implies that $\frac{ B'\paren*{ \mathtt{x}_{ B' }\paren*{ v } } }{ \mathtt{x}_{ B' }\paren*{ v } } \leq \nabla \leq p$. From \cref{lemma:typesoftypes:abc,lemma:typesoftypes:d}, this is possible only if $\vab\paren*{ B' } = \emptyset$, showing the first part. For the second part, from \cref{claim:typesoftypes:ihelper} on $B'$ that it is non-trivial only if $\mathtt{x}_{ B' }\paren*{ v } < 1$ for all $v \geq 0$. However, this would mean that $\nabla' = \nabla$ and the second part follows.

For the third part, we show that any buyer pays at least as much to Bob when his mechanism is $B'$ as he pays when his mechanism is $B$. Fix a buyer with valuation $v \geq 0$. As $\vab\paren*{ B' } = \vab\paren*{ B } = \emptyset$, we have the buyer pays $B'\paren*{ \mathtt{x}_{ B' }\paren*{ v } }$ to Bob when his mechanism is $B'$ and $B\paren*{ \mathtt{x}_B\paren*{ v } }$ when his mechanism is $B$. Using the fact that $\mathtt{x}_B\paren*{ v } < 1$ and $\mathtt{x}_{ B' }\paren*{ v } \in \set*{ \mathtt{x}_B\paren*{ v }, 1 }$, we get that $B'\paren*{ \mathtt{x}_{ B' }\paren*{ v } } \in \set*{ B\paren*{ \mathtt{x}_B\paren*{ v } }, \nabla }$. Now, either $\mathtt{x}_B\paren*{ v } = 0$ implying $B\paren*{ \mathtt{x}_B\paren*{ v } } = 0$ and we are done or $\nabla \geq \frac{ B\paren*{ \mathtt{x}_B\paren*{ v } } }{ \mathtt{x}_B\paren*{ v } } \geq B\paren*{ \mathtt{x}_B\paren*{ v } }$, and we are done.

\end{proof}

\subsection{An Auxiliary Distribution}
\label{sec:auxdistapp}

In this section, we show some lemmas about our auxiliary distribution $D_s$ defined in \cref{eq:fs}, for an arbitrary $s \geq p$. In particular, we prove \cref{lemma:oneseller}. Throughout, $f_s$ and $F_s$ and probability density and the cumulative probability functions for $D_s$ and $\Gamma_s$ is the corresponding revenue curve. 

\begin{lemma}
\label{lemma:fs}
Let $s \geq p$. We have:
\[
F_s\paren*{ v } = \begin{cases}
z \cdot \paren*{ 1 - F(s) } + F\paren*{ v } , &\text{~if~} v \leq p \\
z \cdot \paren*{ 1 - F(s) } + F\paren*{ \frac{ v - a }{ 1 - z } } , &\text{~if~} p < v < a + \paren*{ 1 - z } \cdot s \\
z \cdot \paren*{ 1 - F(s) } + F(s) , &\text{~if~} a + \paren*{ 1 - z } \cdot s \leq v < s \\
z + \paren*{ 1 - z } \cdot F\paren*{ v } , &\text{~if~} s \leq v < \infty \\
\end{cases} .
\]
\end{lemma}
\begin{proof}
We only show the second case as the others are direct from \cref{eq:fs}. For $p \leq v < a + \paren*{ 1 - z } \cdot s$, we have:
\begin{align*}
F_s\paren*{ v } &= z \cdot \paren*{ 1 - F(s) } + F\paren*{ p } + \frac{1}{ 1 - z } \cdot \int_p^v f\paren*{ \frac{ y - a }{ 1 - z } } \diff y \\
&= z \cdot \paren*{ 1 - F(s) } + F\paren*{ p } + \int_p^{ \frac{ v - a }{ 1 - z } } f\paren*{ y } \diff y \tag{Reparametrizing $y \to \frac{ y - a }{ 1 - z }$} \\
&= z \cdot \paren*{ 1 - F(s) } + F\paren*{ \frac{ v - a }{ 1 - z } } .
\end{align*}
\end{proof}

\begin{lemma}
\label{lemma:gammas}
Let $s \geq p$. We have:
\[
\Gamma_s\paren*{ v } = \begin{cases}
v \cdot \paren*{ 1 - z \cdot \paren*{ 1 - F(s) } - F\paren*{ v } } , &\text{~if~} v \leq p \\
v \cdot \paren*{ 1 - z \cdot \paren*{ 1 - F(s) } - F\paren*{ \frac{ v - a }{ 1 - z } } } , &\text{~if~} p < v < a + \paren*{ 1 - z } \cdot s \\
v \cdot \paren*{ 1 - z } \cdot \paren*{ 1 - F(s) } , &\text{~if~} a + \paren*{ 1 - z } \cdot s \leq v < s \\
v \cdot \paren*{ 1 - z } \cdot \paren*{ 1 - F\paren*{ v } } , &\text{~if~} s \leq v < \infty \\
\end{cases} .
\]
\end{lemma}
\begin{proof}
Straightforward from the formula $\Gamma_s\paren*{ v } = v \cdot \paren*{ 1 - F_s\paren*{ v } }$.
\end{proof}

\begin{corollary}
\label{cor:gammas1}
Let $v \geq s \geq s' \geq p$. We have $\Gamma_s\paren*{ v } = \Gamma_{ s' }\paren*{ v }$.
\end{corollary}

\begin{corollary}
\label{cor:gammas2}
Let $s \geq s' \geq p \geq v$. We have $\Gamma_{ s' }\paren*{ v } \leq \Gamma_s\paren*{ v }$.
\end{corollary}

We are now ready to prove \cref{lemma:oneseller}.
\begin{proof}[Proof of \cref{lemma:oneseller}]
We only show that $\rev_B\paren*{ D_s } = \rev_{ \mathsf{B}, \paren*{ A, B } }\paren*{ D }$ as \cref{eq:oneseller} was already argued in \cref{sec:auxdist}. For a buyer with valuation $v \geq 0$, let $r\paren*{ v }$ be the payment of this buyer to Bob when there are two sellers, Alice and Bob, using the mechanisms $A$ and $B$ respectively. As his payment when Bob is the only seller is $B\paren*{ \mathtt{z}_B(v) }$, we have to show that $\int_0^{ \infty } f_s\paren*{ v } \cdot B\paren*{ \mathtt{z}_B\paren*{ v } } \diff v = \int_0^{ \infty } f\paren*{ v } \cdot r\paren*{ v } \diff v$. Using \cref{eq:fs} and reparametrizing $\frac{ v - a }{ 1 - z } \to v$, we have to show that:
\begin{multline*}
\int_0^{ \infty } f\paren*{ v } \cdot r\paren*{ v } \diff v = \int_0^p f\paren*{ v } \cdot B\paren*{ \mathtt{z}_B\paren*{ v } } \diff v + \int_p^s f\paren*{ v } \cdot B\paren*{ \mathtt{z}_B\paren*{ a + \paren*{ 1 - z } \cdot v } } \diff v \\
+ \paren*{ 1 - z } \cdot \int_s^{ \infty } f\paren*{ v } \cdot B\paren*{ \mathtt{z}_B\paren*{ v } } \diff v .
\end{multline*}
To show this, we show that the functions being integrated are pointwise identical. For this, we fix $v \geq 0$. Recall that, in the duopolist setting, if $v \leq s$, we have $v \in \vba\paren*{ B }$ and the buyer purchases from Bob first. This means that $r\paren*{ v } = B\paren*{ \mathtt{x}_B(v) }$ and we are done using \cref{eq:oneseller}. On the other hand, if $v > s$, we have $v \in \vba\paren*{ B }$ and the buyer purchases from Bob second. This means that $r\paren*{ v } = \paren*{ 1 - z } \cdot B\paren*{ \mathtt{x}_B(v) }$ and we are done from \cref{eq:oneseller}.
\end{proof}

\subsection{Fixed Price Mechanisms}
\label{sec:fixedprice}

In this section, we show some lemmas about what happens when Bob uses a fixed price mechanism in our duopolist setting. Recall the notion of a fixed price mechanism from \cref{def:fixedprice}.

\begin{lemma}
\label{lemma:abstartfixed}
Let $q \geq 0$ We have:
\[
\abstart\paren*{ \fp_q } = \begin{cases}
\infty , &\text{~if~} q \leq p \\
q , &\text{~if~} q > p
\end{cases} .
\]
\end{lemma}
\begin{proof}
We prove both cases separately. For the first one, suppose for the sake of contradiction that $q \leq p$ and $v \in \vab\paren*{ \fp_q }$. Observe that we must have $q \leq v$ as otherwise the buyer gets negative utility. By our tie-breaking, this means that $\utilba\paren*{ \fp_q, 1, v } < \utilab\paren*{ \fp_q, \mathtt{x}_{ \fp_q }\paren*{ v }, v }$. Expanding using \cref{eq:utilba,eq:utilab,def:fixedprice} we get:
\[
v - q < z v - a + \paren*{ 1 - z } \cdot \paren*{ v - q } \cdot \mathtt{x}_{ \fp_q }(v) .
\]
As $q \leq v$, we can upper bound $\mathtt{x}_{ \fp_q }(v)$ by $1$ to get $p < q$, a contradiction. Now, consider the case $q > p$. In order to show that $\abstart\paren*{ \fp_q } = q$, we show that for all $v < q$, we have $v \in \vba\paren*{ \fp_q }$ and for all $v > q$, we have $v \in \vab\paren*{ \fp_q }$. For the former, observe that if $v < q \in \vab\paren*{ \fp_q }$, then our tie-breaking implies that this buyer must buy from Bob with positive probability. However, as $v < q$ this would mean that he can get better utility by not buying from Bob, a contradiction. For the latter, assume for the sake of contradiction that $v > q \in \vba\paren*{ \fp_q }$. This means that $\utilab\paren*{ \fp_q, 1, v } < \utilba\paren*{ \fp_q, \mathtt{x}_{ \fp_q }\paren*{ v }, v }$. Expanding using \cref{eq:utilba,eq:utilab,def:fixedprice} we get:
\[
z v - a + \paren*{ 1 - z } \cdot \paren*{ v - q } \leq \paren*{ v - q } \cdot \mathtt{x}_{ \fp_q }(v) + \paren*{ 1 - \mathtt{x}_{ \fp_q }(v) } \cdot \paren*{ z v - a } .
\]
This simplifies to $\paren*{ 1 - z } \cdot \paren*{ v - q } \leq \mathtt{x}_{ \fp_q }(v) \cdot \paren*{ v - q - z \cdot \paren*{ v - p } }$. As the left hand side is positive and $v > q > p$, this is a contradiction.
\end{proof}

\begin{lemma}
\label{lemma:revfixedprice}
Let $q \geq 0$. We have:
\[
\rev_{ \mathsf{B}, \paren*{ A, \fp_q } }\paren*{ D } = \begin{cases}
\Gamma_D(q) , &\text{~if~} q \leq p \\
\paren*{ 1 - z } \cdot \Gamma_D(q) , &\text{~if~} q > p 
\end{cases} .
\]
\end{lemma}
\begin{proof}
Note from \cref{lemma:oneseller,def:fixedprice} that $\rev_{ \mathsf{B}, \paren*{ A, \fp_q } }\paren*{ D } = \rev_{ \fp_q }\paren*{ D_{ \abstart\paren*{ \fp_q } } } = \Gamma_{ \abstart\paren*{ \fp_q } }(q)$. Now, if $q \leq p$, we have from \cref{lemma:abstartfixed} that $\abstart\paren*{ \fp_q } = \infty$ implying that $\Gamma_{ \abstart\paren*{ \fp_q } }(q) = \Gamma_{ \infty }(q) = \Gamma_D(q)$ by \cref{lemma:gammas}. On the other hand, if $q > p$, we have from \cref{lemma:abstartfixed} that $\abstart\paren*{ \fp_q } = q$ implying that $\Gamma_{ \abstart\paren*{ \fp_q } }(q) = \Gamma_q(q) = \paren*{ 1 - z } \cdot \Gamma_D(q)$ by \cref{lemma:gammas}.
\end{proof}

\subsection{A Revenue Inequality for a Single Seller}
\label{sec:myerson-app}

In this section, we state and prove revenue properties for the single seller, single buyer setting with our auxilliary distribution, under certain conditions on the original distribution $D$.

%
\begin{lemma}
\label{lemma:nobottom}
Let $D$ be a distribution, $M$ be a mechanism as above, and $v^* \geq 0$ be given. Assume that $\mathtt{z}_M\paren*{ v^* } > 0$ and define\footnote{Observe that this definition implies that $q \leq v^*$.} $q = \frac{ M\paren*{ \mathtt{z}_M\paren*{ v^* } } }{ \mathtt{z}_M\paren*{ v^* } }$ and $M'\paren*{ z } = \max\paren*{ M\paren*{ z }, q z }$. Suppose that we have:
\begin{align*}
\forall 0 \leq v \leq q: \hspace{1cm} v \cdot f(v) - \paren*{ 1 - F(v) } &\leq q \cdot f(q) - \paren*{ 1 - F(q) } , \\
\forall q \leq v < v^*: \hspace{1cm} q \cdot f(q) - \paren*{ 1 - F(q) } &\leq v \cdot f(v) - \paren*{ 1 - F(v) } .
\end{align*}
It holds that
\[
\rev_M\paren*{ D } \leq \rev_{ M' }\paren*{ D } .
\]
The same also holds with $v^* = \infty$.
\end{lemma}
\begin{proof}
The proof of this lemma has two important claims. The first claim shows that going from the mechanism $M$ to the mechanism $M'$ does not affect the behavior of the buyers with valuation $v > v^*$ and the second claim explains the behavior of buyers with valuation $v \leq v^*$ when the mechanism is $M'$. Formally, we have:
\begin{claim}
\label{claim:nobottom:1}
For all $v \geq v^*$, we have $\mathtt{z}_M\paren*{ v } = \mathtt{z}_{ M' }\paren*{ v }$ and $M\paren*{ \mathtt{z}_M\paren*{ v } } = M'\paren*{ \mathtt{z}_{ M' }\paren*{ v } }$.
\end{claim}
\begin{claim}
\label{claim:nobottom:2}
For all $0 \leq v \leq v^*$, we have $\mathtt{z}_{ M' }(v) = \mathtt{z}_M\paren*{ v^* } \cdot \1\paren*{ v \geq q }$.
\end{claim}
Assuming these proofs for now, we are ready to prove \cref{lemma:nobottom}. We have:
\begin{align*}
\rev_M\paren*{ D } - \rev_{ M' }\paren*{ D } &= \int_0^{ \infty } \paren*{ \mathtt{z}_M(v) - \mathtt{z}_{ M' }(v) } \cdot \paren*{ v \cdot f(v) - \paren*{ 1 - F(v) } } \diff v \tag{\cref{prop:revvv}} \\
&= \int_0^{ v^* } \paren*{ \mathtt{z}_M(v) - \mathtt{z}_{ M' }(v) } \cdot \paren*{ v \cdot f(v) - \paren*{ 1 - F(v) } } \diff v \tag{\cref{claim:nobottom:1}} \\
&= \int_0^{ v^* } \paren*{ \mathtt{z}_M(v) - \mathtt{z}_M\paren*{ v^* } \cdot \1\paren*{ v \geq q } } \cdot \paren*{ v \cdot f(v) - \paren*{ 1 - F(v) } } \diff v \tag{\cref{claim:nobottom:2}} \\
&\leq \paren*{ q \cdot f(q) - \paren*{ 1 - F(q) } } \cdot \int_0^{ v^* } \paren*{ \mathtt{z}_M(v) - \mathtt{z}_M\paren*{ v^* } \cdot \1\paren*{ v \geq q } } \diff v \tag{As $\mathtt{z}_M(v) \leq \mathtt{z}_M\paren*{ v^* }$ by \cref{prop:revvv}} .
\end{align*}
To finish, recall from the Myerson characterization of incentive compatible payments that the utility of a buyer with type $v^*$ facing the mechanism $M$ is just $\int_0^{ v^* } \mathtt{z}_M(v) \diff v$. This means that $\int_0^{ v^* } \mathtt{z}_M(v) \diff v = v^* \cdot \mathtt{z}_M\paren*{ v^* } - M\paren*{ \mathtt{z}_M\paren*{ v^* } } = \mathtt{z}_M\paren*{ v^* } \cdot \paren*{ v^* - q }$ by the definition of $q$. Thus, the last integral in the derivation above is $0$ and we are done. 
\end{proof}

We now prove \cref{claim:nobottom:1,claim:nobottom:2}. 
\begin{proof}[Proof of \cref{claim:nobottom:1}]
Note that, for all $z \in [0, 1]$, we have:
\begin{align*}
z v^* - M\paren*{ z } &\leq v^* \cdot \mathtt{z}_M\paren*{ v^* } - M\paren*{ \mathtt{z}_M\paren*{ v^* } } , \\
z v - M\paren*{ z } &\leq v \cdot \mathtt{z}_M(v) - M\paren*{ \mathtt{z}_M(v) } , \\
z v - M'\paren*{ z } &\leq v \cdot \mathtt{z}_{ M' }(v) - M'\paren*{ \mathtt{z}_{ M' }(v) } .
\end{align*}
Setting $z = \mathtt{z}_M(v)$ in the first equation, and using the definition of $q$, we have $v^* \cdot \mathtt{z}_M(v) - M\paren*{ \mathtt{z}_M(v) } \leq \mathtt{z}_M\paren*{ v^* } \cdot \paren*{ v^* - q }$. Now, as $v^* \leq v$, we have from \cref{prop:revvv} that $0 < \mathtt{z}_M\paren*{ v^* } \leq \mathtt{z}_M\paren*{ v }$. As $q \leq v^*$, we can plug this in and get $q \leq \frac{ M\paren*{ \mathtt{z}_M(v) } }{ \mathtt{z}_M(v) }$. By definition of $M'$, this means that $M'\paren*{ \mathtt{z}_M\paren*{ v } } = M\paren*{ \mathtt{z}_M\paren*{ v } }$. Now, setting $z = \mathtt{z}_{ M' }(v)$ in the second equation and $z = \mathtt{z}_M(v)$ in the third equation, we get:
\begin{align*}
v \cdot \mathtt{z}_{ M' }(v) - M\paren*{ \mathtt{z}_{ M' }(v) } &\leq v \cdot \mathtt{z}_M(v) - M\paren*{ \mathtt{z}_M(v) } , \\
v \cdot \mathtt{z}_M(v) - M'\paren*{ \mathtt{z}_M(v) } &\leq v \cdot \mathtt{z}_{ M' }(v) - M'\paren*{ \mathtt{z}_{ M' }(v) } .
\end{align*}
As $M'\paren*{ \mathtt{z}_M(v) } = M\paren*{ \mathtt{z}_M(v) }$ and $M\paren*{ \mathtt{z}_{ M' }(v) } \leq M'\paren*{ \mathtt{z}_{ M' }(v) }$ (by definition of $M'$), we get:
\[
v \cdot \mathtt{z}_{ M' }(v) - M\paren*{ \mathtt{z}_{ M' }(v) } \leq v \cdot \mathtt{z}_M(v) - M\paren*{ \mathtt{z}_M(v) } \leq v \cdot \mathtt{z}_{ M' }(v) - M\paren*{ \mathtt{z}_{ M' }(v) } .
\]
Thus, both our equation must be tight. With consistent tie-breaking, this is possible only if $\mathtt{z}_M\paren*{ v } = \mathtt{z}_{ M' }\paren*{ v }$, as desired, implying the first part of the claim. The other part of the lemma also follows as we showed that $M'\paren*{ \mathtt{z}_M\paren*{ v } } = M\paren*{ \mathtt{z}_M\paren*{ v } }$.
\end{proof}
\begin{proof}[Proof of \cref{claim:nobottom:2}]
For all $v < q$ and all $z \in [0, 1]$, we have $z v - M'(z) \leq z \cdot \paren*{ v - q } \leq 0$. Moreover, the inequality is strict for all $z > 0$. By definition, we then get $\mathtt{z}_{ M' }(v) = 0$, and the claim follows for all $v < q$. Consider now an arbitrary $q \leq v \leq v^*$ and note by the definition of $\mathtt{z}_{ M' }(v)$ that, for all $z \in [0, 1]$:
\[
z v - M'\paren*{ z } \leq v \cdot \mathtt{z}_{ M' }(v) - M'\paren*{ \mathtt{z}_{ M' }(v) } .
\]
Setting $z = \mathtt{z}_M\paren*{ v^* }$ and using \cref{claim:nobottom:1} and the defintion of $q$ and $M'$, we get $\mathtt{z}_{ M' }\paren*{ v^* } \cdot \paren*{ v - q } \leq \mathtt{z}_{ M' }(v) \cdot \paren*{ v - q }$. From \cref{prop:revvv}, we get that this must hold with an equality implying by our tie-breaking that $\mathtt{z}_{ M' }(v) = \mathtt{z}_M\paren*{ v^* }$, as desired.
\end{proof}

\subsection{Missing Proofs from \texorpdfstring{\cref{sec:oneoptionbob}}{}}

We now complete the proof that, regardless of the distribution $D$, if Bob is restricted to respond with a fixed price mechanism, the bound of $\frac{1}{ \mathrm{e} }$ in \cref{thm:tight2-onebob} is tight. Recall that $\mathcal{M} = \Gamma_D(\texttt{v})$ is Myerson's revenue for the distribution $D$ and Alice's mechanism $A$ is defined as:
\begin{equation}
\label{eq:alice-onebob}
A\paren*{ x } = \begin{cases}
\int_0^x \Gamma_D^{ -1 }\paren*{ \frac{ \mathcal{M} }{ \mathrm{e} \cdot \paren*{ 1 - z } } } \diff z , &\text{~if~} x \leq 1 - \frac{1}{ \mathrm{e} } \\
\infty , &\text{~otherwise~} 
\end{cases} .
\end{equation}

Let $q \geq 0$ and recall that when Bob is using a fixed price mechanism, the buyer always goes to Alice first. This means that, with the mechanism $A$ above for Alice and the fixed price mechanism $\fp_q$ for Bob, if a buyer with type $v$ purchases the option $x$ from Alice he gets a utility of:
\begin{equation}
\label{eq:ab-onebob}
\utilab\paren*{ A, q, x, v } = \begin{cases}
x v - A\paren*{ x }, &\text{~if~} v < q \\
v - q + x q - A\paren*{ x } , &\text{~if~} v \geq q 
\end{cases} .
\end{equation}
The buyer will choose the option $x$ that maximizes this utility. Observe that this maximizer is just $\mathtt{z}_A\paren*{ \min\paren*{ v, q } }$, where $\mathtt{z}_A\paren*{ \cdot }$ is as defined in \cref{sec:prelim}.

Having discussed the behavior of the buyer, we look at the problem from Bob's perspective to find out which value of $q$ maximizes his revenue given that Alice's mechanism is $A$. For this, observe from \cref{eq:ab-onebob} that Bob does not get any revenue from buyers with valuation $v < q$, and from buyer with valuation $v \geq q$, Bob gets revenue $\paren*{ 1 - \mathtt{z}_A\paren*{ q } } \cdot q$. Thus, Bob's expected revenue for a given $q$ equals:
\begin{equation}
\label{eq:bobrev-onebob}
\rev_{ \mathsf{B}, \paren*{ A, \fp_q } }\paren*{ D } = \paren*{ 1 - F\paren*{ q } } \cdot \paren*{ 1 - \mathtt{z}_A\paren*{ q } } \cdot q = \Gamma_D\paren*{ q } \cdot \paren*{ 1 - \mathtt{z}_A\paren*{ q } } .
\end{equation}
It follows that Bob will choose the value $q$ that maximizes the above expression. To compute this value, we first compute the function $\mathtt{z}_A\paren*{ \cdot }$.

\begin{lemma}
\label{lemma:za-onebob}
Let $A$ be as in \cref{eq:alice-onebob}. For all $v \geq 0$, we have:
\[
\mathtt{z}_A\paren*{ v } = \begin{cases}
0 , &\text{~if~} v < \Gamma_D^{ -1 }\paren*{ \frac{ \mathcal{M} }{ \mathrm{e} } } \\
1 - \frac{ \mathcal{M} }{ \mathrm{e} \cdot \Gamma_D\paren*{ v } } , &\text{~if~} \Gamma_D^{ -1 }\paren*{ \frac{ \mathcal{M} }{ \mathrm{e} } } \leq v \leq \Gamma_D^{ -1 }\paren*{ \mathcal{M} } \\
1 - \frac{1}{ \mathrm{e} } , &\text{~if~} \Gamma_D^{ -1 }\paren*{ \mathcal{M} } < v
\end{cases} .
\]
\end{lemma}
\begin{proof}
Recall that $\mathtt{z}_A\paren*{ v }  = \argmax_x \paren*{ x v - A\paren*{ x } }$ and observe from \cref{eq:alice-onebob} that $0 \leq \mathtt{z}_A\paren*{ v } \leq 1 - \frac{1}{ \mathrm{e} }$. Next, note that $x v - A\paren*{ x }$ is increasing at $x$ if and only if $A'\paren*{ x } \leq v$ which is equivalent to $\Gamma_D^{ -1 }\paren*{ \frac{ \mathcal{M} }{ \mathrm{e} \cdot \paren*{ 1 - x } } } \leq v$ by \cref{eq:alice-onebob} for $0 \leq x \leq 1 - \frac{1}{ \mathrm{e} }$. We now argue each part in turn. Observe that the function $\Gamma_D^{ -1 }\paren*{ \cdot }$ is monotone increasing

For the first case, when $0 \leq v < \Gamma_D^{ -1 }\paren*{ \frac{ \mathcal{M} }{ \mathrm{e} } }$, we have that $A'\paren*{ x } > v$ for all $0 \leq x \leq 1 - \frac{1}{ \mathrm{e} }$ and thus $\mathtt{z}_A\paren*{ v } = 0$. For the second case, when $\Gamma_D^{ -1 }\paren*{ \frac{ \mathcal{M} }{ \mathrm{e} } } \leq v \leq \Gamma_D^{ -1 }\paren*{ \mathcal{M} }$, we have that $A'\paren*{ x } > v$ if and only if $x > 1 - \frac{ \mathcal{M} }{ \mathrm{e} \cdot \Gamma_D\paren*{ v } }$. It follows that $\mathtt{z}_A\paren*{ v } = 1 - \frac{ \mathcal{M} }{ \mathrm{e} \cdot \Gamma_D\paren*{ v } }$. Finally, for the third case, when $v > \Gamma_D^{ -1 }\paren*{ \mathcal{M} }$, we have that $A'\paren*{ x } \leq v$ for all $0 \leq x \leq 1 - \frac{1}{ \mathrm{e} }$ and thus $\mathtt{z}_A\paren*{ v } = 1 - \frac{1}{ \mathrm{e} }$.
\end{proof}

\paragraph{Computing Bob's best response.} \cref{lemma:za-onebob} allows us to compute Bob's revenue for every posted price mechanism.

\begin{lemma}
\label{lemma:bobrev-onebob}
For all $q \geq 0$, we have:
\[
\rev_{ \mathsf{B}, \paren*{ A, \fp_q } }\paren*{ D } = \begin{cases}
\Gamma_D\paren*{ q }, &\text{~if~} q < \Gamma_D^{ -1 }\paren*{ \frac{ \mathcal{M} }{ \mathrm{e} } } \\
\frac{ \mathcal{M} }{ \mathrm{e} }, &\text{~if~} \Gamma_D^{ -1 }\paren*{ \frac{ \mathcal{M} }{ \mathrm{e} } } \leq q \leq \Gamma_D^{ -1 }\paren*{ \mathcal{M} } \\
\Gamma_D\paren*{ q } \cdot \frac{1}{ \mathrm{e} } , &\text{~if~} \Gamma_D^{ -1 }\paren*{ \mathcal{M} } < q
\end{cases} .
\]
\end{lemma}
\begin{proof}
Follows from \cref{eq:bobrev-onebob,lemma:za-onebob}.
\end{proof}

\paragraph{Computing Alice's revenue when her mechanism is $A$.} Due to \cref{lemma:bobrev-onebob}, we can assume without loss of generality that Bob's best response is when he sets $q = \Gamma_D^{ -1 }\paren*{ \mathcal{M} }$. We now compute Alice's revenue when Bob uses the mechanism $\fp_q$ with this value of $q$.

\begin{lemma}
\label{lemma:alicerev-onebob}
Let $q = \Gamma_D^{ -1 }\paren*{ \mathcal{M} }$. We have:
\[
\rev_{ \mathsf{A}, \paren*{ A, \fp_q } }\paren*{ D } = \frac{ \mathcal{M} }{ \mathrm{e} } .
\]
\end{lemma}
\begin{proof}
Note from \cref{lemma:za-onebob} that $\mathtt{z}_A\paren*{ q } = 1 - \frac{1}{ \mathrm{e} }$. Observe from \cref{eq:ab-onebob} that Alice's revenue for buyers with valuations $v < q$ is $A\paren*{ \mathtt{z}_A\paren*{ v } }$. For buyers with valuation $v \geq q$, Alice gets revenue $A\paren*{ \mathtt{z}_A\paren*{ q } }$. Thus, Alice's expected revenue when the buyer's values are sampled from $D$ equals:
\begin{align*}
\rev_{ \mathsf{A}, \paren*{ A, \fp_q } }\paren*{ D } &= A\paren*{ \mathtt{z}_A\paren*{ q } } \cdot \paren*{ 1 - F\paren*{ q } } + \int_0^q f\paren*{ v } \cdot A\paren*{ \mathtt{z}_A\paren*{ v } } \diff v \\
&= A\paren*{ \mathtt{z}_A\paren*{ q } } \cdot \paren*{ 1 - F\paren*{ q } } + \int_{ \Gamma_D^{ -1 }\paren*{ \frac{ \mathcal{M} }{ \mathrm{e} } } }^q f\paren*{ v } \cdot A\paren*{ 1 - \frac{ \mathcal{M} }{ \mathrm{e} \cdot \Gamma_D\paren*{ v } } } \diff v \tag{\cref{lemma:za-onebob}} \\
&= A\paren*{ \mathtt{z}_A\paren*{ q } } \cdot \paren*{ 1 - F\paren*{ q } } + \int_{ \Gamma_D^{ -1 }\paren*{ \frac{ \mathcal{M} }{ \mathrm{e} } } }^q \int_0^{ 1 - \frac{ \mathcal{M} }{ \mathrm{e} \cdot \Gamma_D\paren*{ v } } } f\paren*{ v } \cdot \Gamma_D^{ -1 }\paren*{ \frac{ \mathcal{M} }{ \mathrm{e} \cdot \paren*{ 1 - z } } } \diff z \diff v \tag{\cref{eq:alice-onebob}} \\
&= A\paren*{ \mathtt{z}_A\paren*{ q } } \cdot \paren*{ 1 - F\paren*{ q } } + \int_0^{ \mathtt{z}_A\paren*{ q } } \int_{ \Gamma_D^{ -1 }\paren*{ \frac{ \mathcal{M} }{ \mathrm{e} \cdot \paren*{ 1 - z } } } }^q f\paren*{ v } \cdot \Gamma_D^{ -1 }\paren*{ \frac{ \mathcal{M} }{ \mathrm{e} \cdot \paren*{ 1 - z } } } \diff v \diff z \tag{Switching integrals} \\
&= A\paren*{ \mathtt{z}_A\paren*{ q } } \cdot \paren*{ 1 - F\paren*{ q } } + \int_0^{ \mathtt{z}_A\paren*{ q } } \Gamma_D^{ -1 }\paren*{ \frac{ \mathcal{M} }{ \mathrm{e} \cdot \paren*{ 1 - z } } } \cdot \paren*{ F\paren*{ q } - F\paren*{ \Gamma_D^{ -1 }\paren*{ \frac{ \mathcal{M} }{ \mathrm{e} \cdot \paren*{ 1 - z } } } } } \diff z \\
&= A\paren*{ \mathtt{z}_A\paren*{ q } } \cdot \paren*{ 1 - F\paren*{ q } } - \int_0^{ \mathtt{z}_A\paren*{ q } } \Gamma_D^{ -1 }\paren*{ \frac{ \mathcal{M} }{ \mathrm{e} \cdot \paren*{ 1 - z } } } \cdot \paren*{ 1 - F\paren*{ q } } \diff z + \int_0^{ \mathtt{z}_A\paren*{ q } } \frac{ \mathcal{M} }{ \mathrm{e} \cdot \paren*{ 1 - z } } \diff z \tag{Definition of $\Gamma_D\paren*{ \cdot }$} \\
&= \frac{ \mathcal{M} }{ \mathrm{e} } \tag{\cref{eq:alice-onebob} and $\mathtt{z}_A\paren*{ q } = 1 - \frac{1}{ \mathrm{e} }$} .
\end{align*}
\end{proof}


\section{Missing Proofs from the Nash Setting}
\label{sec:fixedtype-app}

We provide the proof of \cref{claim:nonash-2} that was omitted from \cref{sec:fixedtype}.
\begin{proof}[Proof of \cref{claim:nonash-2}]
From \cref{claim:nonash-1}, we have that $A\paren*{ x_A } > 0$ and $B\paren*{ x_B } > 0$ and $x_A < 1$. It follows that $x_A, x_B > 0$. Next, we prove that $x_B - B\paren*{ x_B } = \max_{ z \in [0, 1] } \paren*{ z - B\paren*{ z } }$. Suppose for the sake of contradiction that there exists $z \in [0, 1]$ such that $x_B - B\paren*{ x_B } < z - B\paren*{ z }$. As we already showed that $x_A < 1$, this implies that $\utilab\paren*{ A, B, x_A, x_B } < \utilab\paren*{ A, B, x_A, z }$, a contradiction to the choice of $x_A$ and $x_B$.

Now, we show that $A\paren*{ x_A } < x_A$ and $B\paren*{ x_B } < x_B$. For this note from $\utilba\paren*{ A, B, x_A, x_B } \leq \utilab\paren*{ A, B, x_A, x_B }$ that $B\paren*{ x_B } < x_B$ implies $A\paren*{ x_A } < x_A$. Thus, we only need to show $B\paren*{ x_B } < x_B$. If we assume for contradiction that $x_B \leq B\paren*{ x_B }$, then we have by $x_B - B\paren*{ x_B } = \max_{ z \in [0, 1] } \paren*{ z - B\paren*{ z } }$ that $z \leq B\paren*{ z }$ for all $z \in [0, 1]$. 
In this case, for all $\delta > 0$, consider the mechanism $A'_{ \delta }$ for Alice where the only lottery offered to the buyer is to purchase the item with probability $1$ at a price of $1 - \delta$. We claim that this implies $\rev_{ \mathsf{A}, \paren*{ A', B } }\paren*{ D } = 1 - \delta$. Indeed, this is direct if the buyer goes to Alice first when she uses the mechanism $A'_{ \delta }$ and if the buyer goes to Bob first, it suffices to show that it chooses the option corresponding to $0$ from Bob. This is because, for all $x'_B > 0$, we have:
\[
\utilba\paren*{ A'_{ \delta }, B, 1, x'_B } = x'_B - B\paren*{ x'_B } + \paren*{ 1 - x'_B } \cdot \delta \leq \paren*{ 1 - x'_B } \cdot \delta < \delta = \utilba\paren*{ A'_{ \delta }, B, 1, 0 } . 
\]

Finally, we prove that $x_A - A\paren*{ x_A } = \max_{ z \in [0, 1] } \paren*{ z - A\paren*{ z } }$. For this, recall that $0 < A\paren*{ x_A } < x_A < 1$ and for all $\delta > 0$ small enough, define the mechanism $A'_{ \delta }$ for Alice where the only lottery offered to the buyer is to purchase the item with probability $x_A + \delta$ at a price of $A\paren*{ x_A } + \delta$. As $A\paren*{ x_A } = \rev_{ \mathsf{A}, \paren*{ A, B } }\paren*{ D } \geq \rev_{ \mathsf{A}, \paren*{ A'_{ \delta }, B } }\paren*{ D }$ for all $\delta > 0$, we have that the buyer goes to Bob first in the mechanism $A'_{ \delta }$ for all $\delta > 0$. This means that for all $\delta > 0$, there exists $z_{ \delta } \in [0, 1]$ such that:
\begin{multline*}
\utilab\paren*{ A'_{ \delta }, B, x_A + \delta, x_B } = x_A - A\paren*{ x_A } + \paren*{ 1 - x_A - \delta } \cdot \paren*{ x_B - B\paren*{ x_B } } \\
\leq z_{ \delta } - B\paren*{ z_{ \delta } } + \paren*{ 1 - z_{ \delta } } \cdot \paren*{ x_A - A\paren*{ x_A } } = \utilba\paren*{ A'_{ \delta }, B, x_A + \delta, z_{ \delta } } .
\end{multline*}
This means that for all $\delta > 0$ there exists $z_{ \delta } \in [0, 1]$ such that:
\begin{align*}
\utilab\paren*{ A, B, x_A, x_B } - \delta &= x_A - A\paren*{ x_A } + \paren*{ 1 - x_A } \cdot \paren*{ x_B - B\paren*{ x_B } } - \delta \\
&\leq x_A - A\paren*{ x_A } + \paren*{ 1 - x_A - \delta } \cdot \paren*{ x_B - B\paren*{ x_B } } \\
&\leq \utilba\paren*{ A, B, x_A, z_{ \delta } } \\
&\leq \max_{ x'_B \in [0, 1] } \utilba\paren*{ A, B, x_A, x'_B } \\
&\leq \utilab\paren*{ A, B, x_A, x_B } \tag{Choice of $x_A$ and $x_B$} .
\end{align*}
As $\delta > 0$ is arbitrary, this is possible only if $\max_{ x'_B \in [0, 1] } \utilba\paren*{ A, B, x_A, x'_B } = \utilab\paren*{ A, B, x_A, x_B }$. Next, observe that the fact that $B$ is a mechanism implies that the maximum in this expression must be attained (at the point $x'_B = \mathtt{z}_B\paren*{ 1 - x_A + A\paren*{ x_A } }$). Letting $x'_B$ denote this value henceforth, we again use the choice of $x_A$ and $x_B$ to get that $\max_{ x'_A \in [0, 1] } \utilba\paren*{ A, B, x'_A, x'_B } \leq \utilab\paren*{ A, B, x_A, x_B } = \utilba\paren*{ A, B, x_A, x'_B }$. Using the definition of $\utilba\paren*{ \cdot }$, we are done unless $x'_B = 1$.

However, if $x'_B = 1$, we have that $\utilab\paren*{ A, B, x_A, x_B } = \utilba\paren*{ A, B, x_A, x'_B } = 1 - B\paren*{ 1 }$. It follows from the choice of $x_A$ and $x_B$ that $1 - B\paren*{ 1 } = \max_{ z \in [0, 1] } \paren*{ z - B\paren*{ z } } = x_B - B\paren*{ x_B }$ implying that $B\paren*{ x_B } \leq B\paren*{ 1 }$. Also, noting that $x_A, B\paren*{ x_B } > 0$, we have for small enough $\delta > 0$ that $\rev_{ \mathsf{B}, \paren*{ A, B } }\paren*{ D } = \paren*{ 1 - x_A } \cdot B\paren*{ x_B } < B\paren*{ x_B } - \delta \leq B\paren*{ 1 } - \delta$. Consider the mechanism $B'$ for Bob where the only lottery offered to the buyer is to purchase the item with probability $1$ at a price of $B\paren*{ 1 } - \delta$. Due to the foregoing inequality, we have that, when Bob uses the mechanism $B'$, the buyer must go to Alice first and pick an option, say $x'_A > 0$, from her before going to Bob and picking the option $1$. Moreover, the option $x'_A$ satisfies
\begin{multline*}
\utilba\paren*{ A, B', x_A, 1 } = 1 - B\paren*{ 1 } + \delta \\
\leq x'_A - A\paren*{ x'_A } + \paren*{ 1 - x'_A } \cdot \paren*{ 1 - B\paren*{ 1 } + \delta } = \utilab\paren*{ A, B', x'_A, 1 } . 
\end{multline*}
As $1 - B\paren*{ 1 } = x_B - B\paren*{ x_B } = \utilab\paren*{ A, B, x_A, x_B } \geq \utilab\paren*{ A, B, x'_A, x_B }$, we get that
\[
x'_A - A\paren*{ x'_A } + \paren*{ 1 - x'_A } \cdot \paren*{ x_B - B\paren*{ x_B } } + \delta \leq x'_A - A\paren*{ x'_A } + \paren*{ 1 - x'_A } \cdot \paren*{ x_B - B\paren*{ x_B } + \delta } . 
\]
This simplifies to $x'_A \cdot \delta \leq 0$, a contradiction.
\end{proof}

\end{document}